\documentclass[journal]{IEEEtran}
\pdfoutput=1
\usepackage{amsthm}
\usepackage{amsmath,amssymb, dsfont,bbm}
\usepackage{epsfig,latexsym,amssymb,amsmath,graphics,color}
\usepackage{graphicx,subfigure}
\usepackage{extarrows}
\usepackage{float}
\usepackage[linesnumbered,ruled]{algorithm2e}
\usepackage{epsfig,latexsym,amssymb,amsmath,graphics}
\usepackage{graphicx}
\usepackage[linesnumbered,ruled]{algorithm2e}
\usepackage{cite,color}
\usepackage{url,epstopdf}

\newtheorem{lemma}{Lemma}

\title{Quantumized Microwave Detection Based on $\Lambda$-Type Three-level Superconducting System: HMM Modeling and Performance Prediction}
\author{Junyu Zhang, Chen Gong, Shangbin Li, Shanchi Wu, Rui Ni, Chengjie Zuo, Jinkang Zhu, Ming Zhao, and Zhengyuan Xu
		\thanks{This work was supported by National Key Research and Development Program of China (Grant No. 2018YFB1801904), Key Program of National Natural Science Foundation of China (Grant No. 61631018), Key Research Program of Frontier Sciences of CAS (Grant No. QYZDY-SSW-JSC003). Junyu Zhang, Chen Gong, Shangbin Li, Shanchi Wu, Chengjie Zuo, Jinkang Zhu, Ming Zhao, and Zhengyuan Xu are with Key Laboratory of Wireless-Optical  Communications, Chinese Academy of Sciences, School of Information Science and Technology, University of Science and Technology of China, Hefei, China. Email: jy970102@mail.ustc.edu.cn, \{cgong821,shbli, wsc, czuo, jkzhu, zhaoming, xuzy\}@ustc.edu.cn.

Rui Ni is with Huawei Technology, Shenzhen, China. Email: raney.nirui@huawei.com.}}

\date{}
\begin{document}
\maketitle{}

\begin{abstract}
We adopt artificial $\Lambda$-type three-level system with superconducting devices for microwave signal detection, where the signal intensity reaches the level of discrete photons instead of continuous waveform. Based on the state transition principles of the three-level system, we propose a statistical model for microwave signal detection. In addition, achievable transmission rates and signal detection based on the proposed statistical models are investigated for low temperature conditions in deep space communication scenario. It is predicted that high sensitivity can be achieved by the proposed system. We further characterize the received signal considering the saturation phonomenon of three-level system, which reveals negligible performance degradation caused by saturation under weak received power regime.
\end{abstract}

\section{INTRODUCTION}
With the development of wireless communication systems, the signal reception and detection under weak power regime has attracted extensive attention from both academia and industrial areas. One application lies in satellite communication under long transmission distance  \cite{DS1,DS2,DS3,DS4,DS5}. In particular, extremely high channel attenuation due to long transmission distance requires large transmission power, large antenna size and high-sensitivity receiver.
Among the three factors, high-sensitivity receiver can lead to reduced transmission power requirements and antenna size, which becomes the most fundamental one in improving the communication performance. Based on the principle of wave-particle duality, under extremely weak electromagnetic field intensity, microwave signal may degenerate from continuous waveforms to discrete photons. Such fact inspires us to adopt microwave photon level detection.

On the other hand, from perspectives of condensed matter physics and quantum physics, Josephson junction \cite{SIS1,SIS2,SIS3,SIS4,hbt} or nitrogen-vacancy center\cite{NV1,NV2} is adopted for microwave photon sensing. Compared with nitrogen-vacancy center, Josephson junction using super-conducting devices has higher sensitivity for the signals with carrier up to Giga Hertz. Moreover, super-conducting devices can significantly reduce the resistance compared with the currently adopted semi-conducting devices, and thus can significantly increase the detection sensitivity. The superconducting circuit is also a good choice to implement the quantum computer bus and quantum nodes\cite{BenjaminHuard1,BlakeJohnson1}. In deep space communication, the thermal noise of space receiver rather than interference becomes a crucial factor on the communication performance. Low thermal noise can be realized due to low temperature in the deep space scenario, which can be down to several Kelvins.

Another way to achieve single photon detection is based on the strong coupling between a single quantum emitter and a one-dimensional photon field \cite{TLevel_t}. Due to the energy dissipation interference between the incident field and the radiation of the emitter, the interaction between the emitter and the photon is greatly enhanced. This opens up the possibility of determining the control of the quantum system by a single photon and the possibility of single-photon level detection. In the field of optics, there have been extensive theoretical and experimental works in the quantum physic area\cite{op_TL1,op_TL2,op_TL3,op_TL4,op_TL5,op_TL7,op_TL8,op_TL9,op_TL10}. In the microwave field, coupling system of two-level systems and quantum harmonic oscillators can interact with microwave photons based on cavity QED (Quantum Electrodynamics) or circuit QED, and can be used for single microwave photon detection and quantum states synthesis\cite{SergeHaroche1,DavidSchuster1,AndrewCleland1}. In particular, a single microwave photon can almost certainly induce a transition in a $\Lambda$-type three-level system composed of a resonator and a superconducting qubit dispersion coupling, and change the energy level of the system \cite{ncom1}. In this case, the driving signal induces the coupling system to produce the Rabi oscillation to design the dressed state. The intensity of the driving signal determines the coupling degree between the resonator level and the superconducting qubit level in the dressed state. Appropriate driving signal intensity makes the four main decay paths of the coupled system have the identical decay rate. We call the three-level system working in $\Lambda$ mode\cite{driveTL,Yamamoto_2014}. In case of impedance matching, the reflected field amplitude of the continuous microwave signal incident on the three-level system is almost zero. In other words, almost each input microwave photon can cause the energy level transition of the three-level system  \cite{TL_down1,TL_down2,TL_down3}, which implies that the three-level system has high efficiency in single-photon detection.

Compared with optical photon detection, the main difficulty of microwave photon detection lies in the low energy per microwave photon, which is approximately five orders of magnitude lower than that of an optical photon. Considering the energy of a single optical photon, the detection in room temperature is straightforward. However, due to the low energy of a single microwave photon, the detection needs to be delicately designed.

The microwave photon detector based on artificial three-level system has several advantages. It adopts coherent quantum dynamics to minimize the energy loss during detection and allows the resonance driver to be quickly reset. In addition, the detection does not require time-shaping for input photons. Finally, it can achieve a high detection efficiency based on single device \cite{ncom1}. Therefore, it is intereseting to explore the achievable transmission rate and signal detection for microwave signal up to photon level based on the three-level system with super-conducting devices.

Currently there are discussions on whether the physical layer has been pipelined. Our work responds to such discussion that the pipeline critically depends on the underlying devices, and new devices may bring new challenges. In order to evaluate the performance of the communication system based on $\Lambda$-type three-level system, we also need to analyze the multiphoton response of the three-level system. This paper theoretically builds a statistical model, analyzes the achievable communication rate, and characterizes the received signal for the three-level system under consideration.

In this paper, we first introduce the three stages of a single microwave photon detector based on an artificial $\Lambda$-type three-level system, and propose a statistical model on the three-level system under microwave photons. Based on the statistical model, we investigate the achievable transmission rate and signal detection, and show that the $\Lambda$ three-level system can be predicted to achieve significant gain over the currently deployed sophisticated system under temperature and bandwidth normalization. Finally, we characterize the received signal considering the saturation phonomenon, which reveals negligible performance degradation caused by saturation under weak received power regime. We also analyze the super-/sub-Poisson characteristics of the received signal.

The remainder of this paper is organized as follows. In Section \ref{sec2}, we introduce the composition and working process of microwave photon detector based on three-level system. In Section \ref{sec3}, we propose a single-photon absorption model for the three-level system and calculate the qubit excitation rate under the Poisson arrival of photons. In Section \ref{sec4}, we propose a hidden Markov model in a three-level system, and simulate the bit error rate and achievable transmission rate. The impact of the three-level saturation phenomenon on communication is analyzed in Section \ref{sec5}. Finally, we conclude this paper in Section \ref{sec6}.

\section{MICROWAVE PHOTON DETECTOR BASED ON ARTIFICIAL $\Lambda$-TYPE THREE-LEVEL SYSTEM}\label{sec2}
We propose an end-to-end communication architecture based on the $\Lambda$-type three-level superconducting system, as shown in Fig. \ref{comsys_fig}. The originating antenna emits microwaves which are attenuated through the wireless channel. The receiving end antenna receives extremely weak microwave signals, so that the energy within a symbol period can be compared with the microwave photon energy $h\nu$. And because of the wave-particle duality, the received signal exhibits the characteristics of microwave photons. The received signal resonates with the three-level system after filtering out noise from other frequency bands. The output signals of the three-level system are amplified and sampled, which are fed into the digital processor.
\begin{figure*}[htbp]
  \setlength{\abovecaptionskip}{-0.2cm} 
  \setlength{\belowcaptionskip}{-2cm}
  \centering
  \includegraphics[width=1.95\columnwidth]{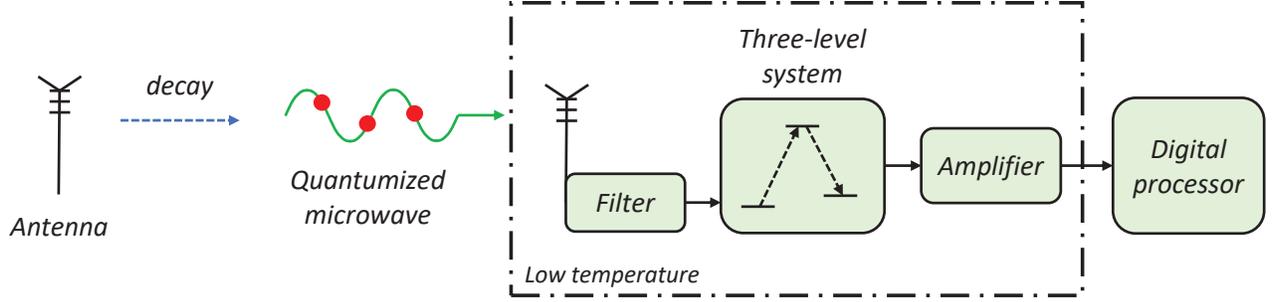}
  \caption{Schematic diagram of communication architecture based on $\Lambda$-type three-level system.}
  \label{comsys_fig}
\end{figure*}

The three-level system under consideration is shown in Fig. \ref{sys_fig}, where the superconducting qubit is dispersively coupled to the transmission line resonator. The resonator is further coupled to a semi-infinite waveguide (WG1) through which the signal photon pulse to be detected is input. WG1 is also adopted to read the qubit and reset the system. Another waveguide (WG2) can apply a driving pulse to the qubit.\cite{TLevel_t}. Assuming on-off keying (OOK) modulation at the transmitter, the detection aims to determine which symbol is transmitted.

\begin{figure}[htbp]
  \setlength{\abovecaptionskip}{-0.2cm} 
  \setlength{\belowcaptionskip}{-2cm}
  \centering
  \includegraphics[width=1\columnwidth]{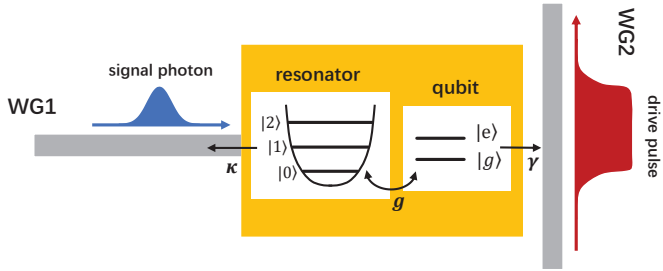}
  \caption{Schematic diagram of the three-level system under consideration (replotted from \cite{TLevel_t}). Symbol $\kappa$ represents the total decay rate of the resonator; symbol $g$ represents the qubit-resonator coupling; and symbol $\gamma$ represents the total decay rate of the qubit.}
  \label{sys_fig}
\end{figure}

As shown in Fig. \ref{stage_fig}, the photon detection consists of three stages: the capture stage, the readout stage, and the reset stage \cite{TLevel_t}. During the capture stage, the superconducting qubit-resonator coupling device enters the $\Lambda$ mode under the driving pulse, where the resonator can transition from the ground state to the excited state. In the readout stage, the energy level of the three-level system is read out through a parametric phase-locked oscillator (PPLO). In the reset phase, a reset pulse is injected from waveguide WG1 to quickly return the three-level system to the ground level. The three stages of the three-level system are alternatively operated, as shown in Fig. \ref{stage_fig}. In Fig. \ref{stage_fig}, we employ $|g\rangle$, $|e\rangle$, and $|u\rangle$ to represent the three states of the three-level system, where $|g\rangle$ denotes the ground state,$|u\rangle$ denotes the intermediate state, $|e\rangle$ denotes the our target state, and the qubit is excited in $|e\rangle$ state.
\begin{figure}[htbp]
  \setlength{\abovecaptionskip}{-0.2cm} 
  \setlength{\belowcaptionskip}{-2cm}
  \centering
  \includegraphics[width=1\columnwidth]{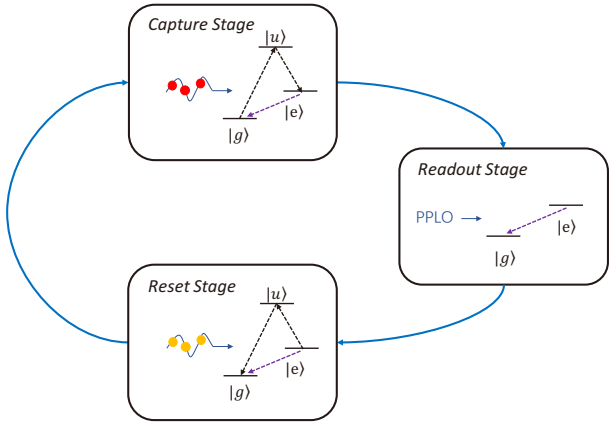}
  \caption{Schematic diagram of the three stages.}
  \label{stage_fig}
\end{figure}

\subsection{Capture of Signal Photons\cite{TLevel_t}}\label{cap}
Driven by the driving pulse, the three-level system enters $\Lambda$ mode from $I$ mode. The two modes are shown in Fig. \ref{mode_fig}. 
When there is a microwave photon input, depending on the photon frequency, the system can transition via $|\tilde{1}\rangle \rightarrow|\tilde{3}\rangle \rightarrow|\tilde{2}\rangle$ or $|\tilde{1}\rangle \rightarrow|\tilde{4}\rangle \rightarrow|\tilde{2}\rangle$ to achieve $|\tilde{1}\rangle$ to $|\tilde{2}\rangle$. If there is no microwave photon input, the system energy level still stays at $|\tilde{1}\rangle$.
\begin{figure}[htbp]
  \setlength{\abovecaptionskip}{-0.2cm} 
  \setlength{\belowcaptionskip}{-2cm}
  \centering
  \includegraphics[width=0.8\columnwidth]{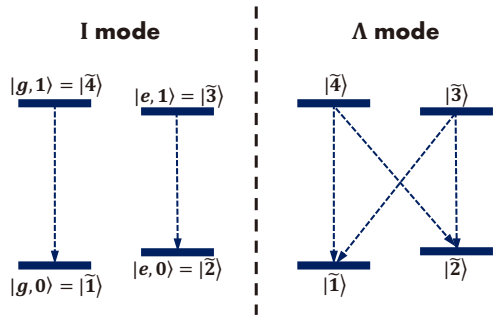}
  \caption{Two modes of three-level system (replotted from \cite{TLevel_t}).}
  \label{mode_fig}
\end{figure}

In the effective three-level system, photon detection is achieved via detecting the energy levels of the superconducting qubit. Energy level $|e\rangle$ means that the input signal is likely to contain microwave photons and the energy level $|g\rangle$ means that the input signal is less likely to contain microwave photons. Measurement should be performed after the driving signal disappears to reduce the influence of the dressed state, improve detection efficiency and reduce the probability of dark count. The amplitude of the driven signal is set to $\Omega_d=\Omega_d^{imd}$ to optimally capture the incident microwave photons. At the optimum driving pluse amplitude, the radiation attenuation rates from $|\tilde{1}\rangle \rightarrow|\tilde{u}\rangle$ and $|\tilde{u}\rangle \rightarrow|\tilde{2}\rangle$ are equal ($u=3$ or $4$). In other words, the three-level system has a high microwave photon absorption rate because of reducing sample reflection and elastic scattering\cite{TL_down3}.

An important factor that affects the detection success rate during the capture phase is the lifetime of the qubits. If the input microwave photon pulse width is too long, the detection efficiency is reduced due to long-term natural decay. When the pulse width of the microwave photon is very short, the detection efficiency of the three-level system is reduced because there is not enough time to make the transition. Under different qubit decay rates $\gamma$, the detection efficiency of different pulse lengths is shown in Fig. \ref{paper_exr_fig}. The smaller the $\gamma$ value, the more sensitive the response of the three-level system to the input signal.

\begin{figure}[htbp]
  \setlength{\abovecaptionskip}{-0.2cm} 
  \setlength{\belowcaptionskip}{-2cm}
  \centering
  \includegraphics[width=1\columnwidth]{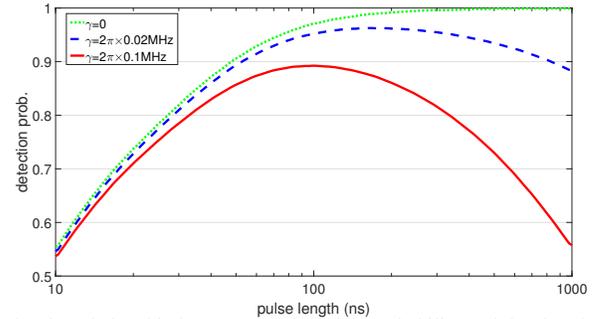}
  \caption{The relationship between the detection probability and the signal pulse length $l$ (replotted from \cite{TLevel_t}).}
  \label{paper_exr_fig}
\end{figure}
\subsection{Readout Stage\cite{TL_Readout}}
In the readout phase, we can use a parametric phase-locked oscillator (PPLO) to distinguish the two energy level states of superconducting qubits $|e\rangle$ and $|g\rangle$. The measurement device and pulse sequence are shown in Fig. 4 in \cite{TL_Readout}. The readout signal is injected into the three-level system, and the reflected wave enters the PPLO for phase lock. The output signal of PPLO is used to read the energy level state. The difference between the reflection phase of the ground state and the excited state of the three-level system is $\pi$ because of the designed readout signal frequency. The circulator in the circuit causes the microwave signal entering any port to be transmitted only to the next port in rotation.

The error in the readout phase is mainly caused by the natural decay of $|2\rangle\rightarrow|1\rangle$ during the measurement process. Assuming that the time used for phase lock is $t_w$, the readout error can be expressed as $1-e^{-\gamma t_w}$. The phase lock error of the phase lock device itself is negligible compared with other errors \cite{TL_Readout}. In other words, the dark count in the readout phase can be ignored.

\subsection{Reset of the System\cite{TLevel_t}}\label{sec_reset}
The reset phase aims to reset the three-level system back to the initial level. Although the system can be reset by relying on the natural decay of the qubits, the reset requires a long time because of the low decay rate. In order to shorten the dead time of the detector, a microwave transition method can be used for quick reset.

A driving signal is injected from waveguide WG2, and a microwave photon pulse with an average photon number $n$ at a specific frequency is injected into the waveguide WG1 to enable the three-level transition $|\tilde{2}\rangle \rightarrow|\tilde{u}\rangle \rightarrow|\tilde{1}\rangle$. 

\section{SIGNAL CHARACTERIZATION UNDER MICROWAVE PHOTON POSSION ARRIVAL}\label{sec3}

\subsection{Statistical Model of Single-Photon Absorption}
We characterize the photon absorption of the three-level system based on the exponential law of energy level transitions. Under the optimal working condition, the reflection coefficient $|r|$ of the three-level system to microwave photons tends to zero. In other words, the incident microwave photons can cause energy level transitions in a three-level system. According to Section \ref{cap}, we assume that the equivalent transition rate of the three-level transition process is $\kappa/4$, the natural decay rate of the qubit is $\gamma$, and the probability distribution of microwave photon arrival time is ${|f_s(t)|} ^2\triangleq\rho(t)$. When $-\left(\beta l+w\right)/2<t<\left(\beta l+w\right)/2$, the driving signal becomes active and the incident photons can be absorbed, and the system observes at time $t_{o}>(\beta l+w) / 2$. The probability of the excited state at the observation time $t_o$, denoted as $\mathbb{P}(|\tilde{2}\rangle | t=t_{o})$, is given by,
\begin{equation}
\mathbb{P}(|\tilde{2}\rangle | t=t_{o})=\int_{-t_i}^{t_i} d t \rho(t) \int_{t}^{t_i} d q \frac{\kappa}{4} e^{-\frac{\kappa}{4}\left(q-t\right)} e^{-\gamma\left(t_{o}-q\right)}
.\label{e4}\end{equation}
Based on Eq. (\ref{e4}), we can get
\begin{equation}\mathbb{P}(|\tilde{2}\rangle | t=t_{o})=\int_{-t_{i}}^{t_{i}} d t \frac{\kappa\rho(t)e^{-\gamma \Delta_{t}}}{\kappa-4 \gamma} \left(1-e^{-\left(\frac{\kappa}{4}-\gamma\right) \Delta_{d}}\right),\end{equation}
where $t_i=\left(\beta l+w\right)/2$, $\Delta_t=t_o-t$ and $\Delta_d=t_i-t$.

If no microwave photons are incident, the false excitation rate $\left.P_{0}=\mathbb{P}(|2\rangle|| \tilde{1}\rangle, t=t_{o}\right)$. Since the modified states $|\tilde{1}\rangle$ and $|\tilde{2}\rangle$ are orthogonal, we have
\begin{equation}\left.\left.\mathbb{P}(|1\rangle|| \tilde{2}\rangle, t=t_{o}\right)=\mathbb{P}(|2\rangle|| \tilde{1}\rangle, t=t_{o}\right).\end{equation}

The detection probability $P_1$ of single photon incidence is
\begin{equation}P_{1}=\left(1-P_{0}\right) \mathbb{P}\left(|\tilde{2}\rangle | t=t_{o}\right)+P_{0}\left(1-\mathbb{P}\left(|\tilde{2}\rangle | t=t_{o}\right)\right).\label{e7}\end{equation}

The detection efficiency of three-level system is simulated based on Eq. (\ref{e7}). The numerical results are shown in Fig. \ref{exr_fig}. We only adopt the exponential statistical law of energy level transitions. It is seen that the results obtained in Fig. \ref{exr_fig} are very similar to those in Fig. \ref{paper_exr_fig} considering the Hamiltonian of the system.
\begin{figure}[htbp]
  \setlength{\abovecaptionskip}{-0.2cm} 
  \setlength{\belowcaptionskip}{-2cm}
  \centering
  \includegraphics[width=1\columnwidth]{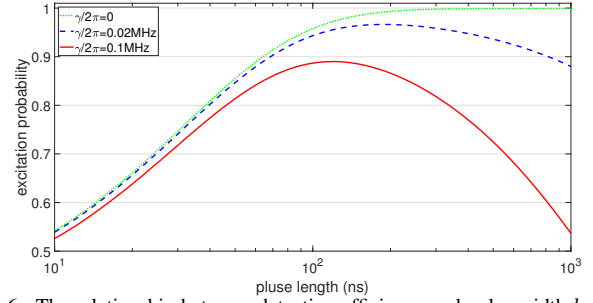}
  \caption{The relationship between detection efficiency and pulse width $l$ under different $\gamma$ values based on statistical model.}
  \label{exr_fig}
\end{figure}

\subsection{Detection Efficiency Under Microwave Photon Poisson Arrival}
Note that the microwave signal exhibits particle properties under extremely weak transmission power, which can be well characterized by Poisson arrival, e.g., a coherent weak microwave signal. Therefore, assuming that the microwave photons conform to Poisson process, we investigate the detection efficiency of the three-level system.

When the arrival time interval of the two photons is shorter than a threshold, the three-level system will obviously saturate\cite{TLevel_t}.  Let $T_c$ denote the duration of the capture phase of the three-level system, and $\bar{n}$ denote the mean number of photons for Poisson arrival within time $T_c$. If the three-level transition time $1 / \kappa \ll \mathrm{T}_{\mathrm{c}} / \bar{n}$, the photon saturation issue can be negligible. In this case, we assume that each photon independently initiates a three-level system transition.

In the case of a Poisson source input with a arrival rate of $\lambda$, assume that microwave photon incidence begins at $t_i=0$, microwave photon incidence ends at $t_f=T_c$, and the observation is performed at $t_o=T_c+\Delta_o$. The probability of excitation at time $t_o$ and $T_c$ under the Poisson arrival intensity $\lambda$, denoted as $\mathbb{P}\left(|\tilde{2}\rangle | \lambda, t=t_{o}\right)$ and $\mathbb{P}(|\tilde{2}\rangle | \lambda, t=T_{c})$, are given by,
\begin{equation}\mathbb{P}\left(|\tilde{2}\rangle | \lambda, t=t_{o}\right)=e^{-\gamma \Delta_{o}} \mathbb{P}\left(|\tilde{2}\rangle | \lambda, t=T_{c}\right),\label{e8}\end{equation}
\begin{equation}\mathbb{P}(|\tilde{2}\rangle | \lambda, t=T_{c})=
\sum_{N=0}^{\infty} \frac{\left(\lambda T_{c}\right)^{N}}{N !e^{\lambda T_{c}}} \mathbb{P}\left(|\tilde{2}\rangle | \lambda, N, t=T_{c}\right),\label{e9}\end{equation}
where $N$ is the number of photons arriving in time $T_c$ and $\mathbb{P}\left(|\tilde{2}\rangle | \lambda, N, t=T_{c}\right)$ is the excitation probability further with given number of photons $N$. It is known from that for Poisson arrival that given the number of arriving photons, the time of arrival of each photon conforms to a uniform distribution, independent from the Poisson distribution arrival rate $\lambda$. Thus, we have
\begin{equation}\mathbb{P}\left(|\tilde{2}\rangle | \lambda, \mathrm{N}, t=T_{c}\right)=\mathbb{P}\left(|\tilde{2}\rangle | \mathrm{N}, t=T_{c}\right),\label{e10}\end{equation}
\begin{equation}\begin{aligned}\mathbb{P}(|\tilde{2}\rangle | N, t=&T_{c})=\\
&\int d s^{N} \frac{N !}{T_{C}^{N}} \mathbb{P}\left(|\tilde{2}\rangle | N, t=T_{c}, s^{N}=\{t^a_i\}\right),\label{e11}\end{aligned}\end{equation}
where $\mathbb{P}\left(|\tilde{2}\rangle | \mathrm{N}, t=T_{c}\right)$ is the excitation probability under uniform distribution of $N$ photons' arrival time, $t_i^a$ is the arrival time of the $i^{th}$ photon and $0\le t_1^a\le t_2^a\le...\le t_N^a\le T_c$. The $\mathbb{P}\left(|\tilde{2}\rangle | N, t=T_{c}, s^{N}=\{t^a_i\}\right)$ is the excitation probability further with given photons arrival time. When a photon arrives, the three-level transition can only be excited when the system energy level is $| \tilde{1}\rangle$. When the three-level system is well initialized, we assume that the photon label causing the three-level transition is given by $K^p=\{k_1,...,k_p\}$, where $p\geq1$ and $1=k_1\le... \le k_p\le N$. The excited state probability under this condition is given by,
\begin{equation}\begin{aligned}
\mathbb{P}(|\tilde{2}\rangle | N, &t=T_{c},s^{N}=\{t^a_i\})=\\
&\sum_{{K\in\cal K}} \mathbb{P}\left(K | N, T_{c},\{t^a_i\}\right) \mathbb{P}\left(|\tilde{2}\rangle | N, T_{c},K,\{t^a_i\}\right),
\label{e12}\end{aligned}\end{equation}
where ${\cal K}$ is the set of all $\{k_1,...,k_p\}$ and $T^N=\left(t_{1}^{a}, \ldots, t_{N}^{a}\right)$.

Due to the neglect of the three-level saturation phenomenon, we assume that the three-level transition meets the principle of incompatibility. In other words, the photon causing the new transition must arrive after the previous photon transition is completed. The probability $f_b(t)$ of the three-level system starting from time 0 and returning to the ground state before time $t$ is
\begin{equation}\begin{aligned}f_{b}(t)&=\int_{0}^{t} d s \frac{\kappa}{4} e^{-\frac{\kappa}{4} s}\left(1-e^{-\gamma(t-s)}\right)\\
&=1-e^{-\frac{\kappa}{4} t}-\frac{\kappa}{\kappa-4 \gamma}\left(e^{-\gamma t}-e^{-\frac{\kappa}{4} t}\right),\label{e13}\end{aligned}\end{equation}
\begin{equation}\begin{aligned}\mathbb{P}(K | N&, T_{c},\{t^a_i\})=\\
&\left(1-f_{b}\left(t_{N}^{a}-t_{k_{p}}^{a}\right)\right) \prod_{i=2}^{p}\left(f_{b}\left(\Delta^g_{i}\right)-f_{b}\left(\Delta^e_{i}\right)\right),\label{e14}\end{aligned}\end{equation}
where $\Delta^g_{i}=t_{k_{i}}^{a}-t_{k_{i-1}}^{a}$ and $\Delta^e_{i}=t_{k_{i}-1}^{a}-t_{k_{i-1}}^{a}$.

Given $\{t_1^a,...,t_N^a\}$ and $\{k_1,...,k_p\}$, the magnitude of the excitation probability depends only on the previous transition time. Let $f_{|\widetilde{2}\rangle}\left(t, t_{0}\right)$ represent the excitation probability that the three-level system starts a three-level transition at $t=0$ and no photons cause a transition in $0<t<t_0$. Let $f_{|\widetilde{1}\rangle}\left(t, t_{0}\right)$ represent the probability that the state of the three-level system is $|\tilde{1}\rangle$ under the same conditions. Then, we can get the following equations,
\begin{equation}\begin{aligned}f_{|\tilde{2}\rangle}\left(t, t_{0}\right)&=\int_{0}^{t} d s \frac{\kappa}{4} e^{-\frac{\kappa}{4} s}\left(e^{-\gamma(t-s)}\right)\\&=\frac{\kappa}{\kappa-4 \gamma}\left(e^{-\gamma t}-e^{-\frac{\kappa}{4} t}\right),\label{e15}\end{aligned}\end{equation}
\begin{equation}\begin{aligned}f_{|\widetilde{1}\rangle}\left(t, t_{0}\right)&=e^{-\frac{\kappa}{4} t}+\int_{t_{0}}^{t} d s \frac{\kappa}{4} e^{-\frac{\kappa}{4} s}\left(1-e^{-\gamma(t-s)}\right)
\\&=e^{-\frac{\kappa}{4} t_{0}}+\frac{\kappa\left(e^{-\frac{\kappa}{4} t_{0}-\gamma(t-t 0)}-e^{-\frac{\kappa}{4} t}\right)}{\kappa-4 \gamma},\label{e16}\end{aligned}\end{equation}
\begin{equation}\begin{aligned}\mathbb{P}&\left(|\tilde{2}\rangle | N, T_{c}, k_{p},\left(t_{k_{p}}^{a}, \ldots, t_{N}^{a}\right)\right)=\\&\frac{f_{|\tilde{2}\rangle}\left(T_{c}-t_{k_{p}}^{a}, t_{N}^{a}-t_{k_{p}}^{a}\right)}{f_{|\tilde{1}\rangle}\left(T_{c}-t_{k_{p}}^{a}, t_{N}^{a}-t_{k_{p}}^{a}\right)+f_{|\tilde{2}\rangle}\left(T_{c}-t_{k_{p}}^{a}, t_{N}^{a}-t_{k_{p}}^{a}\right)}.\label{e17}\end{aligned}\end{equation}

Combining the Eq. (\ref{e8}) to Eq. (\ref{e17}), we can calculate $\mathbb{P}(|\tilde{2}\rangle | \lambda, t=t_{o})$. The detection success rate $p_{\lambda}^c$ in the acquisition phase is given by,
\begin{equation}\begin{aligned}P_{\lambda}^{c}=\left(1-P_{0}\right) \mathbb{P}\left(|\tilde{2}\rangle | \lambda, t=t_{o}\right)+P_{0}\left(1-\mathbb{P}\left(|\tilde{2}\rangle | \lambda, t=t_{o}\right)\right).\end{aligned}\end{equation}
In the readout phase, the probability of PPLO detecting the incident photon is given by
\begin{equation}P_{\lambda}^{o u t}=p_{w} P_{\lambda}^{c},\end{equation}
where $p_w=e^{-\gamma t_w}$ is the successful phase lock in the readout phase.

According to Section \ref{sec_reset}, the probability of reset error is
\begin{equation}P_{\lambda}^{r e}=P_{g}\left(1-P_{\lambda}^{o u t}\right)+P_{e} P_{\lambda}^{o u t}.\end{equation}

\subsection{Numerical Results of Detection Efficiency}
Under $T_c=230ns$, $\Delta_o=35ns$ and $t_w=48ns$, the miss detection probability of photon arrival with different arrival rates, transition rate $\kappa$ and decay rate $\gamma$ is shown in Figure \ref{p2to1_fig}. It can be seen that the excitation rate of qubit rises faster with larger $\kappa$ and smaller $\gamma$. And the qubit excitation rate causes a larger upper limit under the larger $\gamma$. In the case of a reset error, the qubit is initially in an excited state. Under no power input, the system with reset error has a high probability of excited state in the readout phase. Under the correct reset, the response of the detector has a larger dynamic range. Therefore, it has better communication performance than the incorrect reset case.

\begin{figure*}[htbp]
	\centering
	\subfigure[$\gamma=2\pi\times0.1MHz$]{\includegraphics[width=0.98\columnwidth]{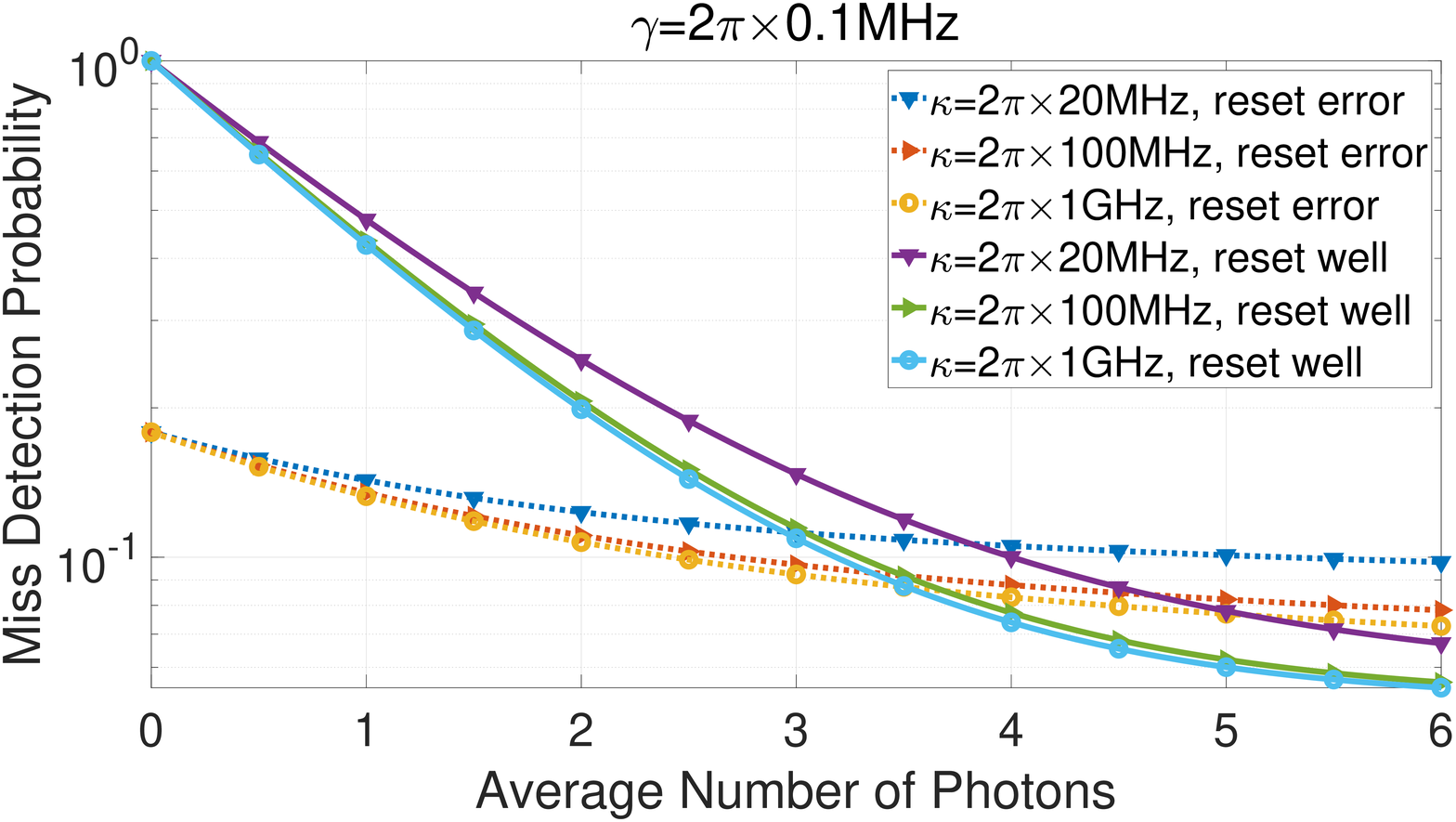}}\quad
	\subfigure[$\gamma=2\pi\times0.2MHz$]{\includegraphics[width=0.98\columnwidth]{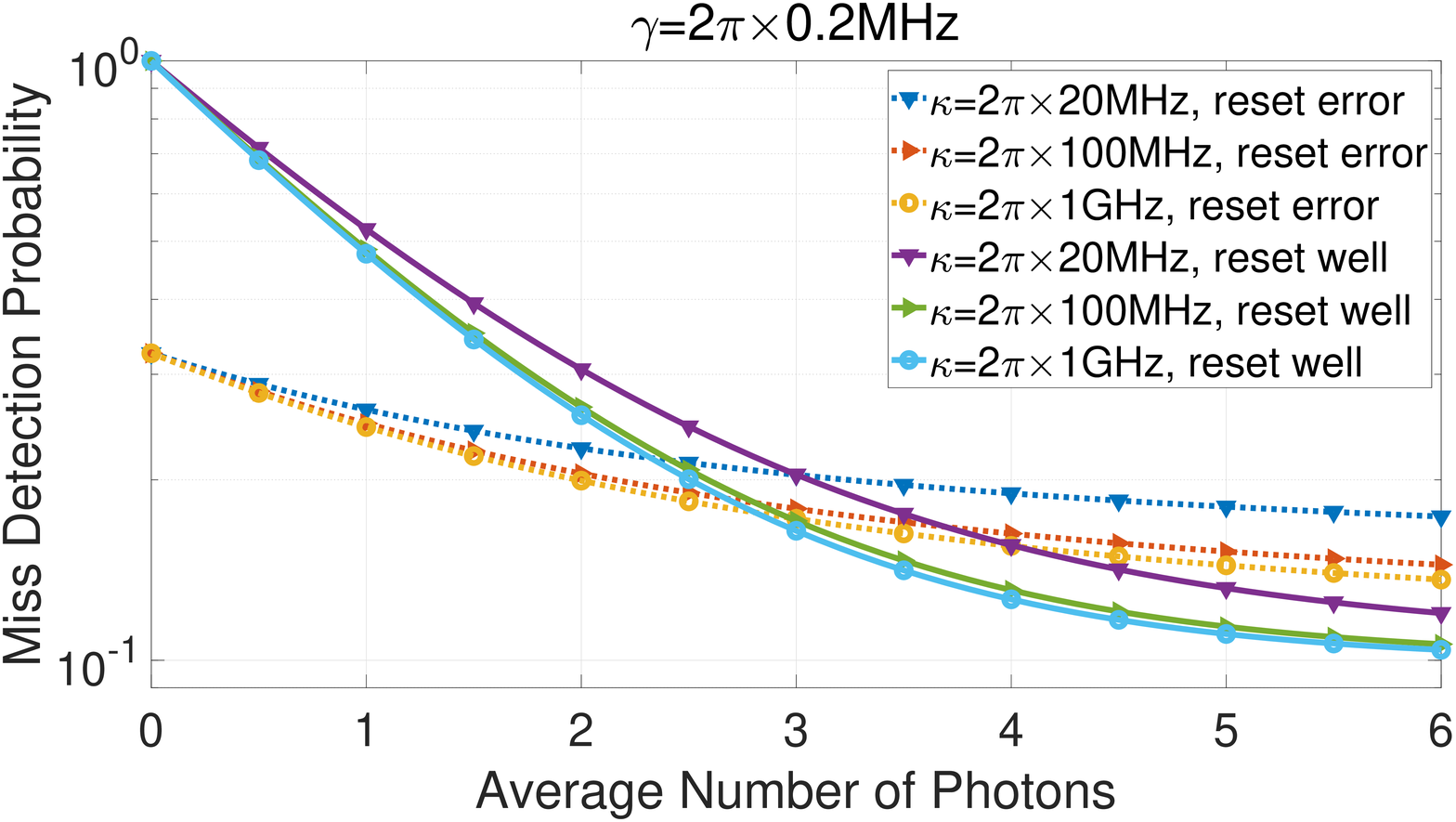}}\\	
	\subfigure[$\gamma=2\pi\times0.4MHz$]{\includegraphics[width=0.98\columnwidth]{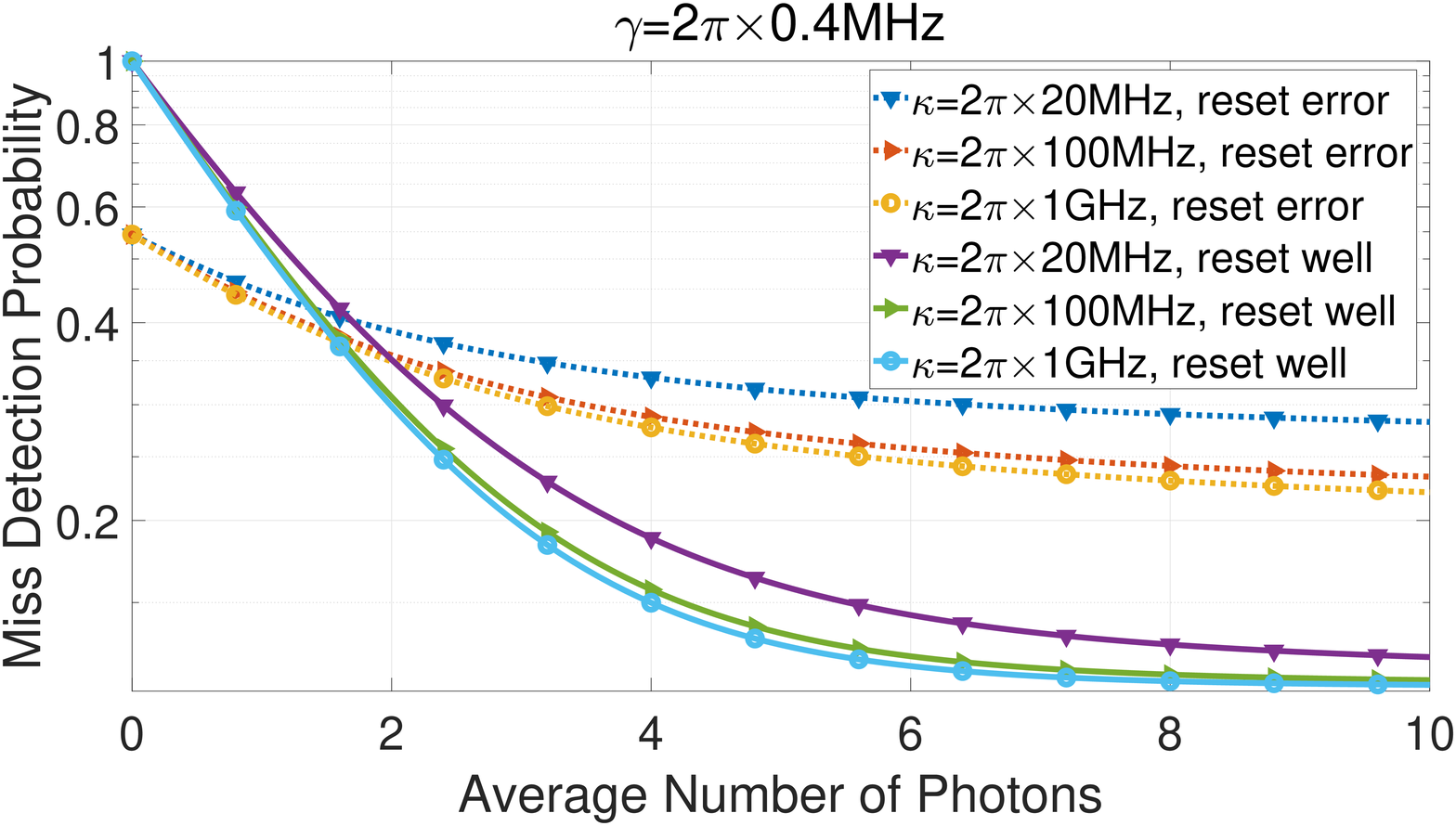}}\quad
	\subfigure[$\gamma=2\pi\times1.0MHz$]{\includegraphics[width=0.98\columnwidth]{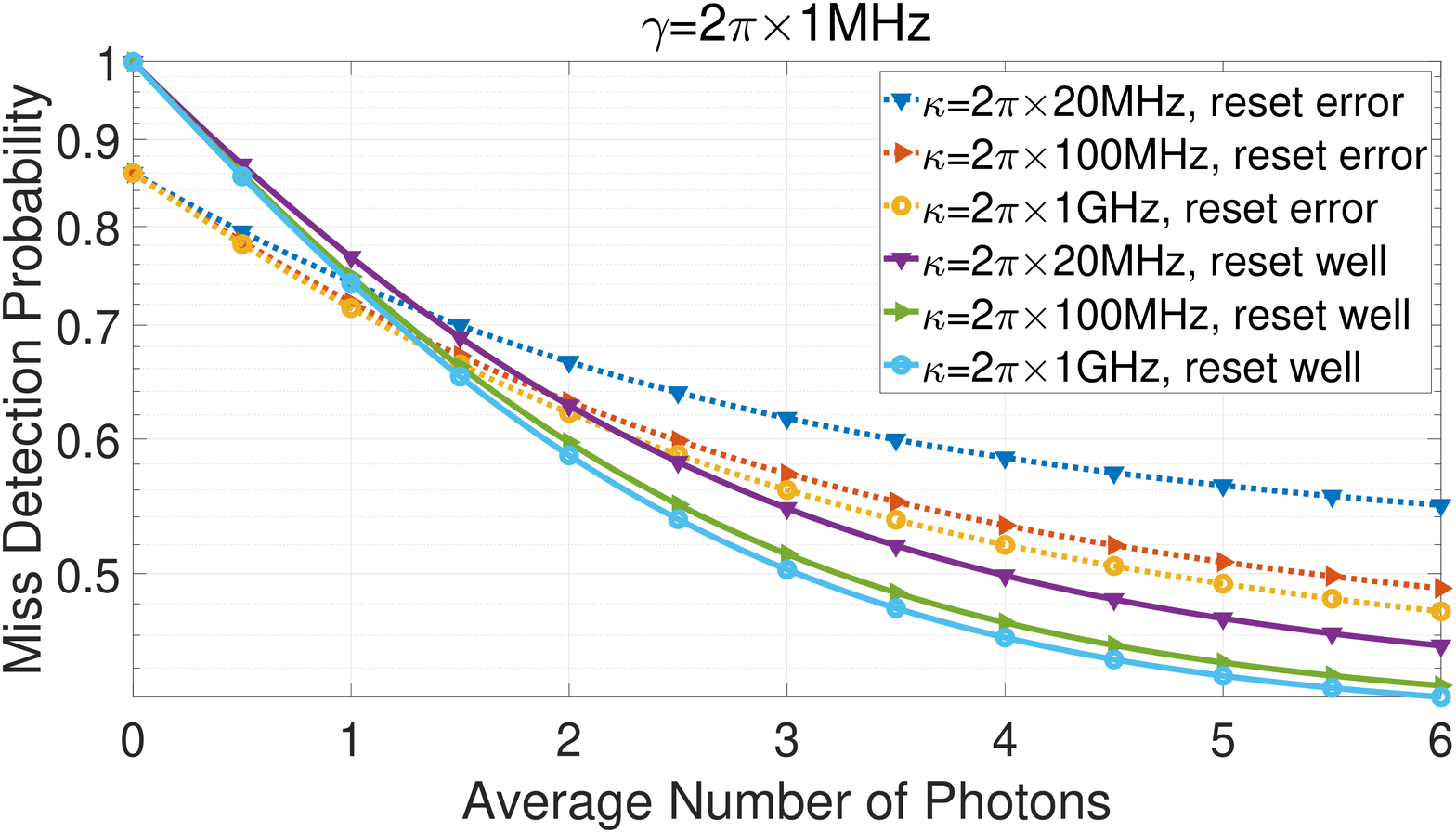}}\\	
	\caption{Miss detection probability at the end of the readout stage.\label{p2to1_fig}}
\end{figure*}

\section{SIGNAL DETECTION BASED ON HIDDEN MARKOV MODEL}\label{sec4}

\subsection{Hidden Markov Model of Three-level System}
Note that the initial state of the three-level system depends on the reset condition at the previous detection. In other words, the detection efficiency of the three-level system is determined by the previous reset and the current signal input. We can only observe the output of PPLO instead of the input microwave signal and the energy level of the three-level system. Therefore, the three-level system can be modeled as a hidden Markov chain in multi-period detection process.

We assume that the photon arrival conforms to the Poisson distribution with OOK modulation. Symbol $0$ means that no signal microwave photons are sent, and symbol $1$ means that the photons arrival rate is $\lambda_1$. We assume that the thermal noise also conforms to the Poisson distribution with arrival rate $n_e$. Assuming that a symbol maintains $N$ three-level detection cycles, the hidden state of the HMM with three-level continuous detection can be characterized as $\left(|i\rangle, S_{j}\right)$, where $|i\rangle$($i=1\ or\ 2$) is the initial energy level at the beginning of the symbol and $S_j$($j=0\ or\ 1$) represents the OOK symbol. The observation quantity is denoted as $o \in\{0,1\}^{N}$, where $0$ represents no photon has been detected, and $1$ represents that the photons have been detected. The state transition diagram of HMM is shown in Fig. \ref{HMM_fig}.
\begin{figure}[htbp]
  \setlength{\abovecaptionskip}{-0.2cm} 
  \setlength{\belowcaptionskip}{-2cm}
  \centering
  \includegraphics[width=0.65\columnwidth]{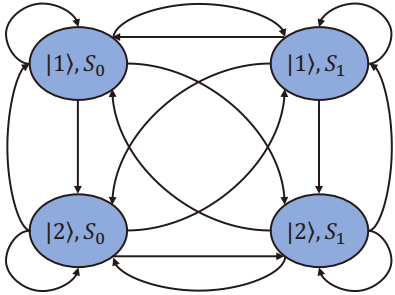}
  \caption{HMM state transition diagram.}
  \label{HMM_fig}
\end{figure}

We assume that two OOK symbols are sent with equal probability and the three-level system is well initialized when starting continuous detection. We can characterize the HMM under the model.

\subsection{Viterbi Decoding for the Three-level System}\label{sec_ber}
Based on the HMM, we adopt Viterbi algorithm for symbol detection\cite{HMM_t}. We assume that the transmitted symbol sequence is $\{S_n\}$ and the decoding results are $\{|\hat{\imath}(n)\rangle, \hat{S}(n)\}$. The symbol error rate can be given by
\begin{equation}\lim _{M \rightarrow \infty} \frac{1}{M} \sum_{n=1}^{M} \operatorname{Pr}(S(n) \neq \hat{S}(n)).\end{equation}

Assuming that the input port has thermal noise $P_N=\kappa_B T_e B$ with a temperature of $T_e$ and the signal bandwidth $B=1/(NT_c)$, where $\kappa_B$ is the Boltzmann constant, the average photon number $n_e$ of thermal noise in time $T_c$ is given by,
\begin{equation}
n_e=\frac{\kappa_B T_e}{Nh\nu}.
\end{equation}

Under the parameters $\gamma=2\pi\times0.1MHz$, $T_c=230ns$, $\Delta_o=35ns$, $t_w=48ns$, same $P_g$ and $P_e$ in Section \ref{sec_reset}, the simulation result of symbol error rate at the different noise temperature $T_e$ and period duration number $N$ is shown in Fig. \ref{ber_fig1}. For a given thermal noise temperature, we adjust the period number $N$, which is inversely proportional to the bandwidth, such that the thermal noise photon number is equivalent to a $10mK$ thermal noise source. According to Fig. \ref{ber_fig1}, a system with a larger $\kappa$ has lower BER because of the more sensitive input signal response. At the same thermal noise level, the system with large $N$ has lower BER. We adopt $\kappa=2\pi\times1GHz$, $T_e=8K$ and $N=800$ to calculate the sensitivity. When $N=800$, the bit rate is $1/(N\times(T_c+\Delta_o+t_w))=1/(800\times(230ns+35ns+48ns))\approx4kbps$. When the power is $-148.3dBm$, the bit error rate reaches $10^{-3}$. According to the OOK bit error rate formula, the SNR margin is $9.8dB$, and thus the normalized sensitivity is $-148.3dBm-9.8dB=-158.1dBm$ for SNR $0dB$. The 4G/5G standard system working under $(4kbps, 8K)$ has a sensitivity gain $10log_{10}(2.2M/4k)dB+10log_{10}(300/8)dB\approx43dB$ when linearly converted from $(2.2Mbps, 300K)$. Considering the LTE signal sensitivity of $-100dBm$ to $-105dBm$ with data rate $2.2Mbps$ for LTE and 5G communication system\cite{equipmentradio}, the sensitivity of 4G/5G system under the same temperature and data rate is $-143dBm$ to $-148dBm$. The sensitivity gain of our proposed structure can reach $10dB$ to $15dB$. Since target BER $10^{-3}$ is considered under low temperature and data rate, we focus on singal power regime below $-145dBm$.
\begin{figure*}[htbp]
  \setlength{\abovecaptionskip}{-0.2cm} 
  \setlength{\belowcaptionskip}{-2cm}
  \centering
  \includegraphics[width=1.95\columnwidth]{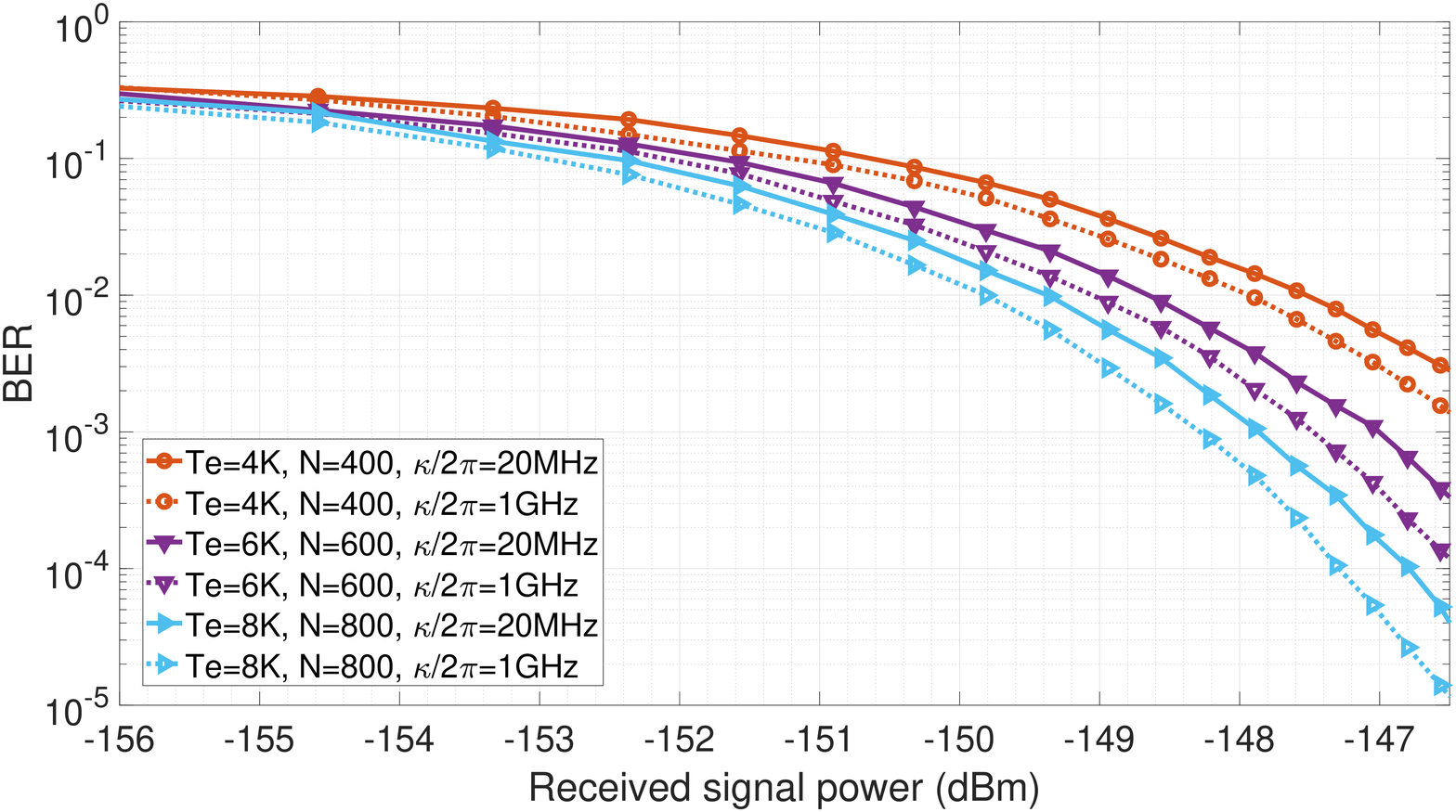}
  \caption{The BER under different received power.}\label{ber_fig1}
\end{figure*}

\subsection{Achievable Transmission Rate}\label{sec_rate}
Assume that the symbol sequence in continuous detection is $S^T=s(1)s(2)...s(T)$ and the output sequence is $O^T=o(1)o(2)...o(T)$. The achievable transmission rate is given by
\begin{equation}\begin{aligned}I&=\lim _{T \rightarrow \infty} \frac{I\left(S^{T}; O^{T}\right)}{T}=\lim_{T \rightarrow \infty} {\frac{H(S^T)}{T}-\frac{H(S^T|O^T)}{T}}\\&=1-\lim_{T \rightarrow \infty} \frac{\sum_{o^t}P(O^T=o^t)H(S^T|O^T=o^t)}{T}.\end{aligned}\end{equation}
We use Monte-Carlo method to simulate the achievable transmission rate \cite{HMM_t}, where the detailed procedure is omitted due to lack of space. Figure \ref{tranC_fig} shows the simulation result of the achievable transmission rate when other parameters are the same as those in Section \ref{sec_ber}.

\begin{figure*}[htbp]
	\centering
	\subfigure[$\gamma=2\pi\times0.1MHz$]{\includegraphics[width=0.98\columnwidth]{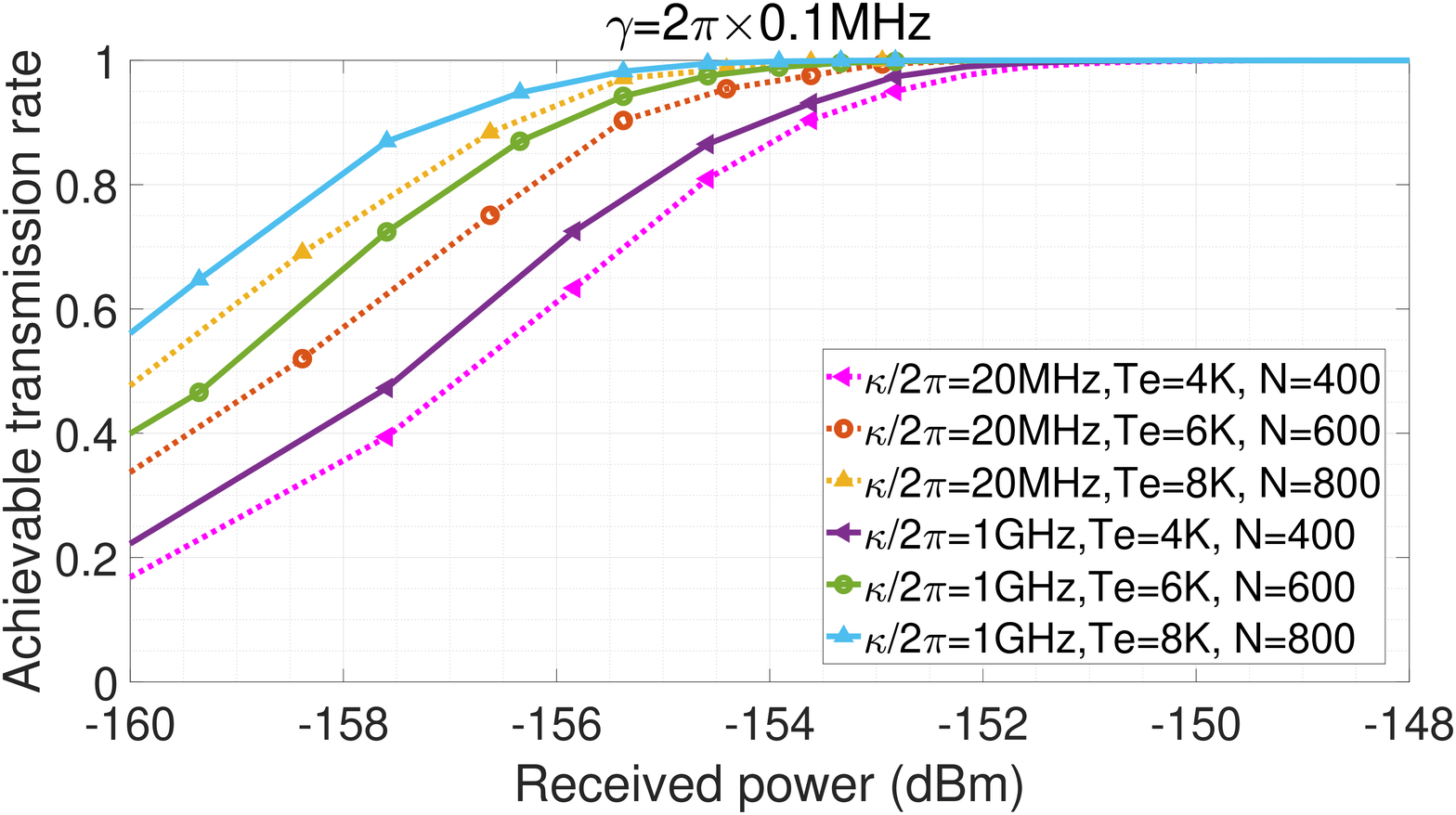}}\quad
	\subfigure[$\gamma=2\pi\times0.2MHz$]{\includegraphics[width=0.98\columnwidth]{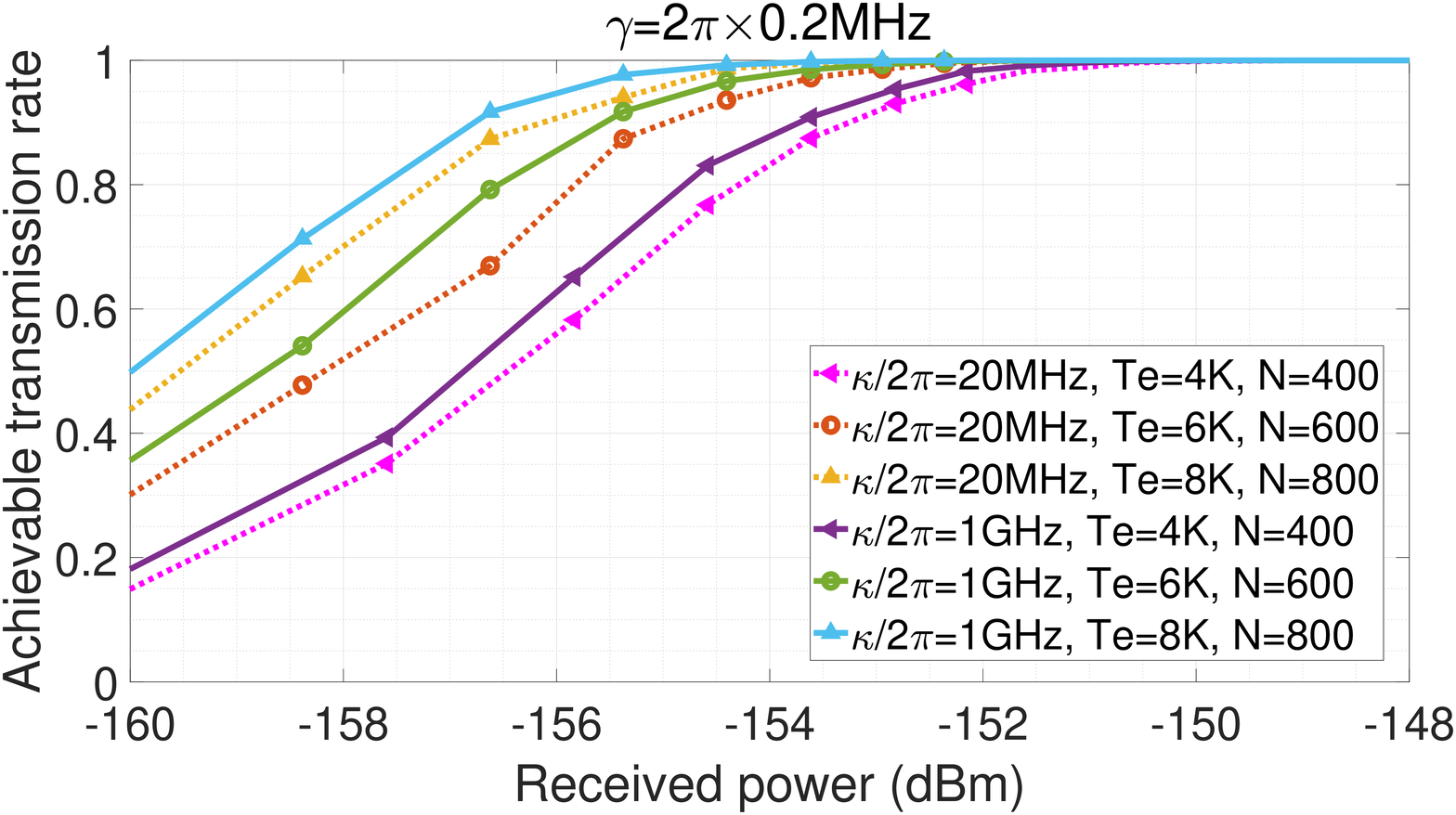}}\\	
	\subfigure[$\gamma=2\pi\times0.4MHz$]{\includegraphics[width=0.98\columnwidth]{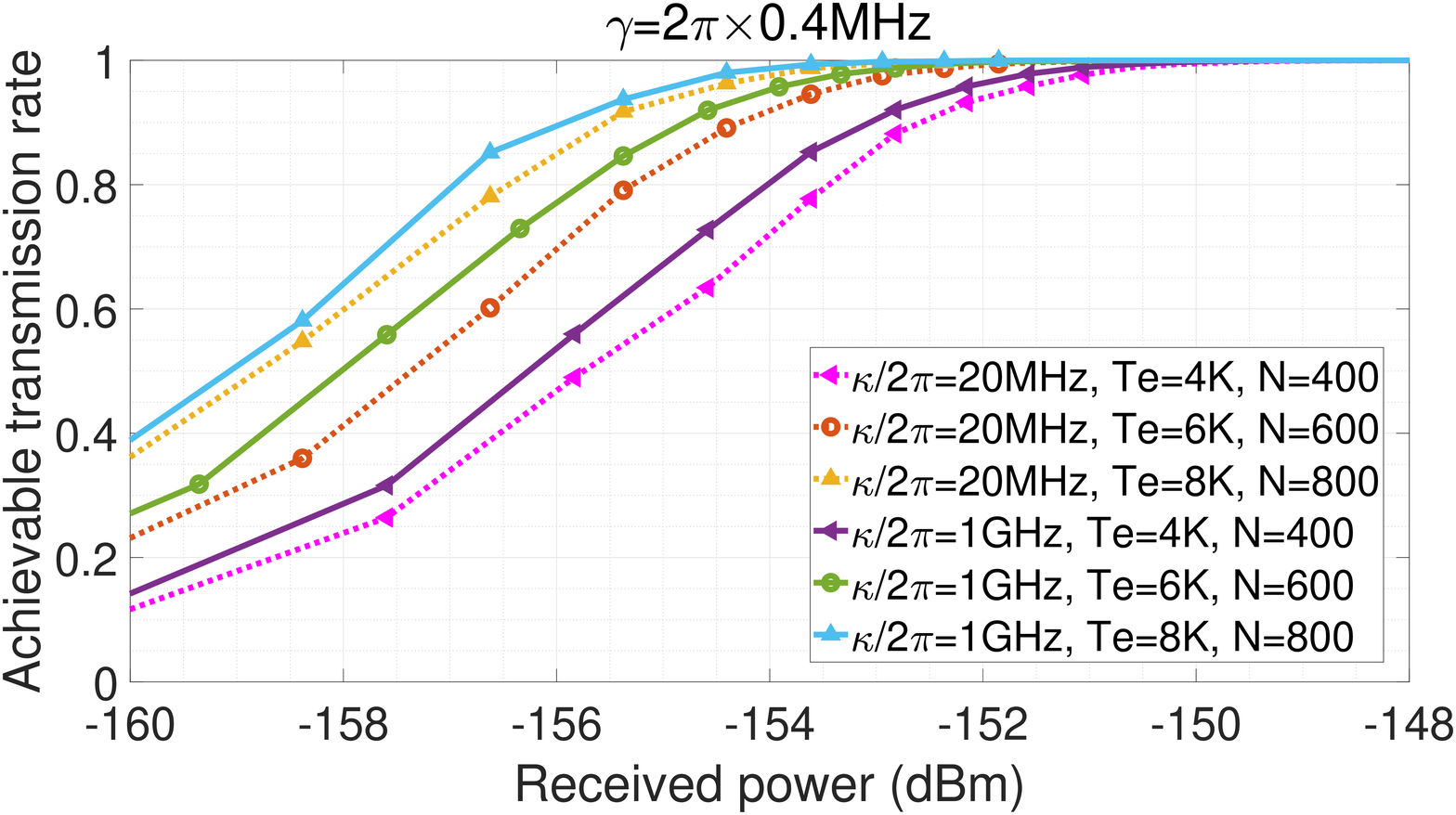}}\quad
	\subfigure[$\gamma=2\pi\times1.0MHz$]{\includegraphics[width=0.98\columnwidth]{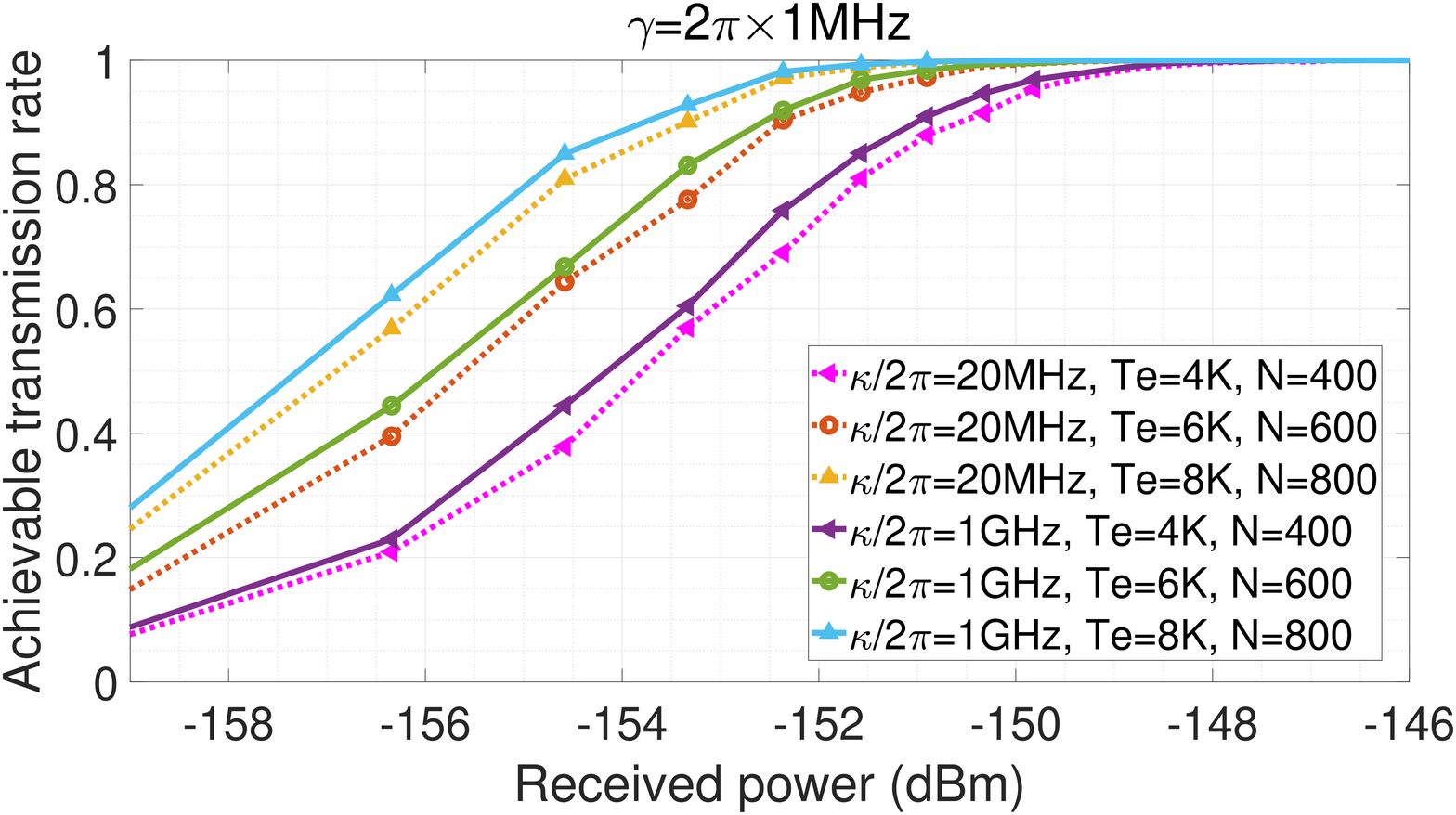}}\\	
	\caption{The achievable transmission rate with respect to the received power, decay rate $\gamma$ and transition rate $\kappa$.}\label{tranC_fig}
\end{figure*}

According to Fig. \ref{tranC_fig}, the system with larger $N$ and $\kappa$ values and smaller $\gamma$ values has better performance. Under $\kappa=2\pi\times1GHz$, $\gamma=2\pi\times0.1MHz$, $T_e=8K$ and $N=800$, when the signal power is $-156.5dBm$, the achievable transmission rate can reach 0.95. Considering the LTE signal sensitivity of $-143dBm$ to $-148dBm$ with data rate $4kbps$ and temperature $8K$ for LTE and 5G communication system under linear conversion, the sensitivity gain of our proposed structure can reach $8dB$ to $13dB$.

\section{SATURATION CHARACTERISTICS OF THREE-LEVEL SYSTEM}\label{sec5}

\subsection{Saturation Model with Fixed Time Window}
The transition rate of a three-level system depends on radiation rate $\kappa$ of the waveguide. If the time interval between the arrival of two microwave photons is very short, the detuning of the three-level system will occur, which may destroy the energy level transition.

In order to characterize the three-level saturation of a large number of photons input, we simplified the interaction between microwave photons and the three-level system. We assume that the arrival of any microwave photon will cause the saturation of the three-level system and destroy the transition, within time $\tau=\alpha\kappa^{-1}$ after the arrival of a microwave photon. To match results in the related work\cite{TLevel_t} under least-squares criterion, we perform an exhaustive search for $\alpha$ and find the optimal $\alpha=1.14$. Compared with the results in related work, the assumption of a fixed time window can well describe the saturation effect of a three-level system, as shown in Fig. \ref{good2to1_fig}.
\begin{figure}[htbp]
  \setlength{\abovecaptionskip}{-0.2cm} 
  \setlength{\belowcaptionskip}{-2cm}
  \centering
  \includegraphics[width=1\columnwidth]{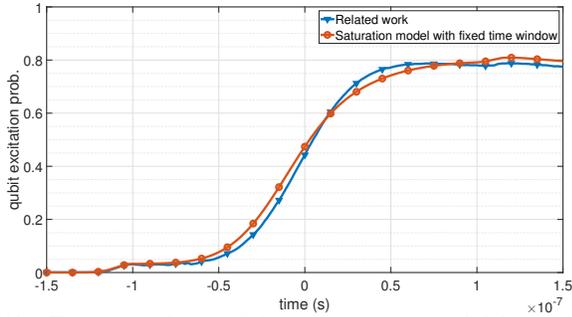}
  \caption{The time evolution of the qubit excitation probability with two-photon input in the related work (replotted from \cite{TLevel_t}) and in the saturation model with a fixed time window.}
  \label{good2to1_fig}
\end{figure}

\subsection{Survival 	Photon Characteristics under Poisson Arrival}
Surviving photons are defined as photons that are not affected by the saturation state. Only the surviving photons can affect the transition of the three-level system. We assume that the capture phase time is $T_c$, and the time window length is $\tau$, and calculate the mean and variance of the number of surviving photons under Poisson arrival. 

We first calculate the first and second moments of the survival photons under the uniform distribution of $N$ photon arrival time. We mark $I_i$ as a sign of whether the $i$-th photon is alive, given by
\begin{equation}I_{i}=\left\{\begin{array}{l}
1,\ i^{th}\ photon\ survives\\
0,\ otherwise
\end{array}\right..\end{equation}
\newtheorem{thm}{Theorem}
\begin{thm}\label{thm1}
If $T/\tau \ge 4$, given the number of incident photons $N$, the first order moment of the number of surviving photons denoted as $E_S(N)$, is given by
\begin{equation}E_S(N)=2\left(1-\frac{\tau}{T_c}\right)^{N}+(N-2)\left(1-\frac{2 \tau}{T_c}\right)^{N}.\end{equation}
Note that the second order moment of the number of surviving photons, denoted as $D_S(N)$, is given by
\begin{equation}\begin{aligned}D_{S}(N)=&2\left(1-\frac{\tau}{T_c}\right)^{N}+(N+4)\left(1-\frac{2 \tau}{T_c}\right)^{N}\\&+\left(N^{2}-7 N+12\right)\left(1-\frac{4 \tau}{T_c}\right)^{N}\\&+(6 N-18)\left(1-\frac{3 \tau}{T_c}\right)^{N}.\end{aligned}\end{equation}
\end{thm}
\begin{proof}
Please refer to Appendix \ref{app1}.
\end{proof}
We assume that the arrival rate of the Poisson distribution is $\lambda$. Based on Theorem \ref{thm1}, we have the following on the number of surviving photons under Poisson arrival.
\begin{lemma} \label{lemma1}
Under the Poisson distribution and the mean value $\Lambda$, we have
\begin{equation}
\mathbb{E}_{\Lambda}[\alpha^N]=e^{-(1-\alpha)\Lambda},
\end{equation}
\begin{equation}
\mathbb{E}_{\Lambda}[N\alpha^N]=\alpha\Lambda e^{-(1-\alpha)\Lambda},
\end{equation}
\begin{equation}
\mathbb{E}_{\Lambda}[N^2\alpha^N]=((\alpha\Lambda)^2+\alpha\Lambda)e^{-(1-\alpha)\Lambda}.
\end{equation}
\end{lemma}
\begin{proof}
We have
\begin{equation}\begin{aligned}
\mathbb{E}_{\Lambda}[\alpha^N]&=\sum_{N=0}^\infty \alpha^N e^{-\Lambda}\frac{\Lambda^N}{N!}\\&=e^{-(1-\alpha)\Lambda}\sum_{N=0}^\infty e^{\alpha\Lambda}\frac{(\alpha\Lambda)^N}{N!} =e^{-(1-\alpha)\Lambda},
\end{aligned}\end{equation}
\begin{equation}
\begin{aligned}
\mathbb{E}_{\Lambda}[N\alpha^N]&=\sum_{N=0}^\infty N\alpha^N e^{-\Lambda}\frac{\Lambda^N}{N!}\\&=e^{-(1-\alpha)\Lambda}\sum_{N=0}^\infty Ne^{\alpha\Lambda}\frac{(\alpha\Lambda)^N}{N!} \\&=e^{-(1-\alpha)\Lambda}\mathbb{E}_{\alpha\Lambda}[N]=\alpha\Lambda e^{-(1-\alpha)\Lambda},
\end{aligned}
\end{equation}
\begin{equation}
\begin{aligned}
\mathbb{E}_{\Lambda}[N^2\alpha^N]&=\sum_{N=0}^\infty N^2\alpha^N e^{-\Lambda}\frac{\Lambda^N}{N!}\\&=e^{-(1-\alpha)\Lambda}\sum_{N=0}^\infty N^2 e^{\alpha\Lambda}\frac{(\alpha\Lambda)^N}{N!} \\&=e^{-(1-\alpha)\Lambda}\mathbb{E}_{\alpha\Lambda}[N^2]\\&=((\alpha\Lambda)^2+\alpha\Lambda) e^{-(1-\alpha)\Lambda}.
\end{aligned}
\end{equation}\end{proof}
\begin{thm}\label{thm2}
Under Poisson arrival and $T_c/\tau\ge4$, the first order and second order moments of the number of surviving photons are respectively
\begin{equation}E_{\lambda}=2 e^{-\lambda \tau}+(\lambda(T_c-2 \tau)-2) e^{-2 \lambda \tau},\end{equation}
\begin{equation}\begin{aligned}
D_{\lambda}=&2 e^{-\lambda \tau}+(\lambda(T_c-2 \tau)+4) e^{-2 \lambda \tau}\\&+(6(\lambda T_c-3 \lambda \tau)-18) e^{-3 \lambda \tau} \\&+\left((\lambda T_c-4 \lambda \tau)^{2}-6(\lambda T_c-4 \lambda \tau)+12\right) e^{-4 \lambda \tau}.
\end{aligned}\end{equation}
The variance of the number of surviving photons is
\begin{equation}\begin{aligned}
\sigma_{\lambda}^{2}&=D_{\lambda}-E_{\lambda}^{2}\\&= 2 e^{-\lambda \tau}+\lambda(T-2 \tau) e^{-2 \lambda \tau}+(2 \lambda T-10 \lambda \tau-10) e^{-3 \lambda \tau} \\
&\quad+\left(-4 \lambda^{2} T \tau+12 \lambda^{2} \tau^{2}-2 \lambda T+16 \lambda \tau+8\right) e^{-4 \lambda \tau}.
\end{aligned}\end{equation}
\end{thm}
\begin{proof}
$E_\lambda$ and $D_\lambda$ can be directly calculated based on Theorem \ref{thm1} and Lemma \ref{lemma1}. And the variance $\delta^2_\lambda$ can also be directly calculated because $\delta^2_\lambda=D_\lambda-E_\lambda^2$.
\end{proof}
An fundamental property is whether the surviving photons have a sub-Poisson or super-Poisson distribution. Let $\Delta_\lambda\triangleq E_{\lambda}-D_{\lambda}$. The survival photon distribution characteristics can be judged by the sign of $\Delta_\lambda$. 

\begin{thm}\label{thm3}
There exists certain $\lambda_0>0$, such that $\Delta_\lambda>0$ for $\lambda\in(0,\lambda_0)$,  $\Delta_\lambda<0$ for $\lambda\in(\lambda_0,+\infty)$ and $\Delta_\lambda=0$ for $\lambda=0,\lambda_0$.
\end{thm}
\begin{proof}
Please refer to Appendix \ref{app2}.
\end{proof}
It can be seen that for lower arrival rate, the surviving photons exhibit a sub-Poisson distribution; and as the arrival rate increases, the survival photon characteristics turns to super-Poisson distribution. The $\Delta_\lambda$ under several $T/\tau$ values is shown in Fig. \ref{subposs_fig}.
\begin{figure}[htbp]
  \setlength{\abovecaptionskip}{-0.2cm} 
  \setlength{\belowcaptionskip}{-2cm}
  \centering
  \includegraphics[width=1\columnwidth]{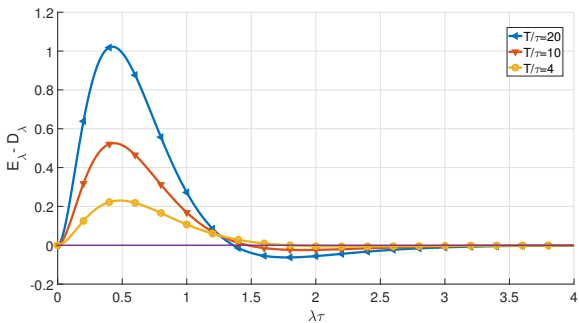}
  \caption{$\Delta_\lambda$ under different $T/\tau$.}
  \label{subposs_fig}
\end{figure}

\subsection{Performance of the Three-level System under Saturation Model}\label{sec_zu_rate}
The excitation probability of qubit will decrease because the large microwave photon arrival rate causes long-term saturation of the three-level system. We calculate the qubit excitation probability under different reset conditions of the three-level system, as shown in Fig. \ref{sa2to1_fig} and Fig. \ref{sa2to1re_fig}.
\begin{figure}[htbp]
  \setlength{\abovecaptionskip}{-0.2cm} 
  \setlength{\belowcaptionskip}{-2cm}
  \centering
  \includegraphics[width=1\columnwidth]{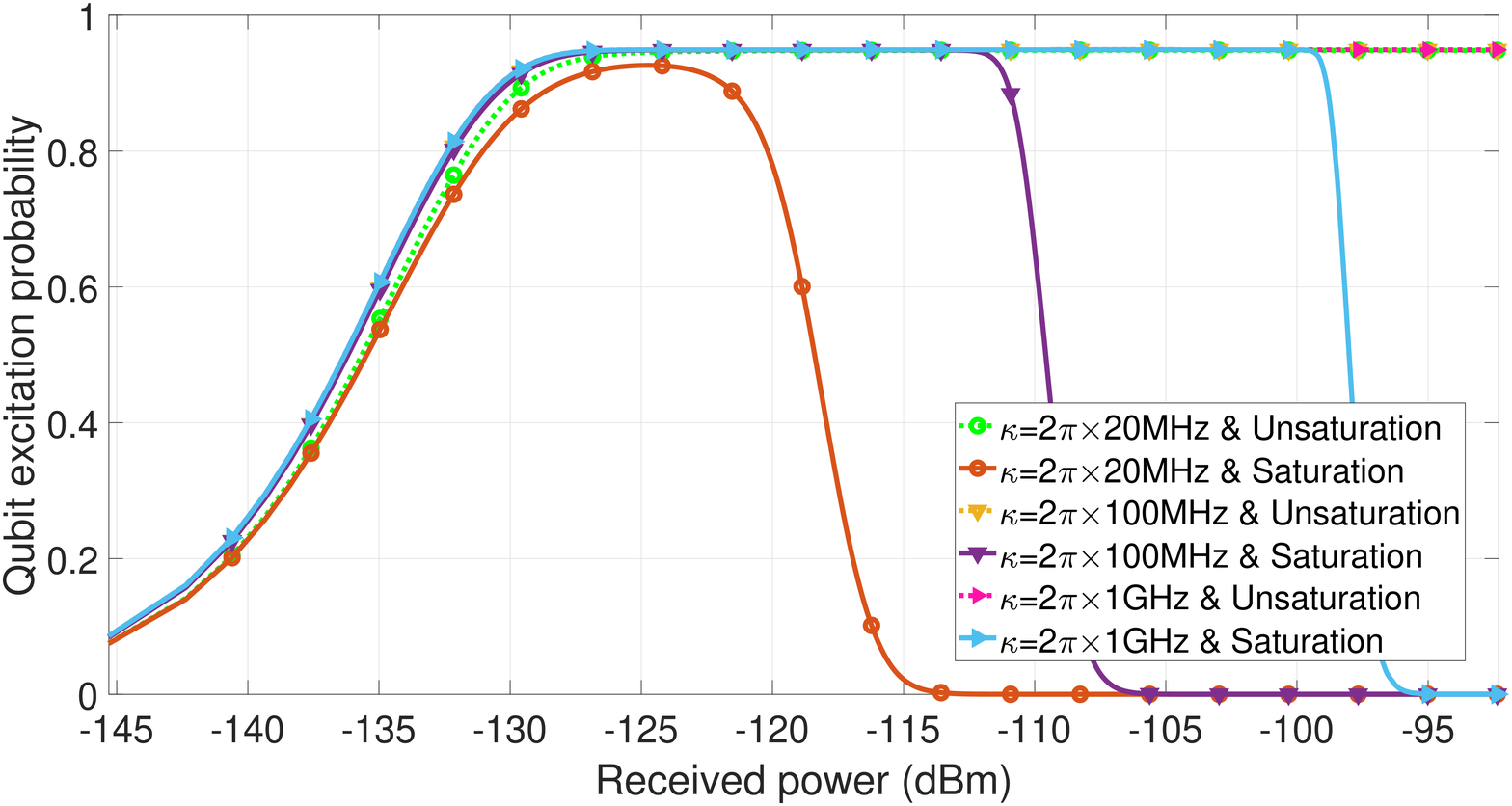}
  \caption{The qubit excitation probability under correct reset and different transition rates $\kappa$.}
  \label{sa2to1_fig}
\end{figure}

\begin{figure}[htbp]
  \setlength{\abovecaptionskip}{-0.2cm} 
  \setlength{\belowcaptionskip}{-2cm}
  \centering
  \includegraphics[width=1\columnwidth]{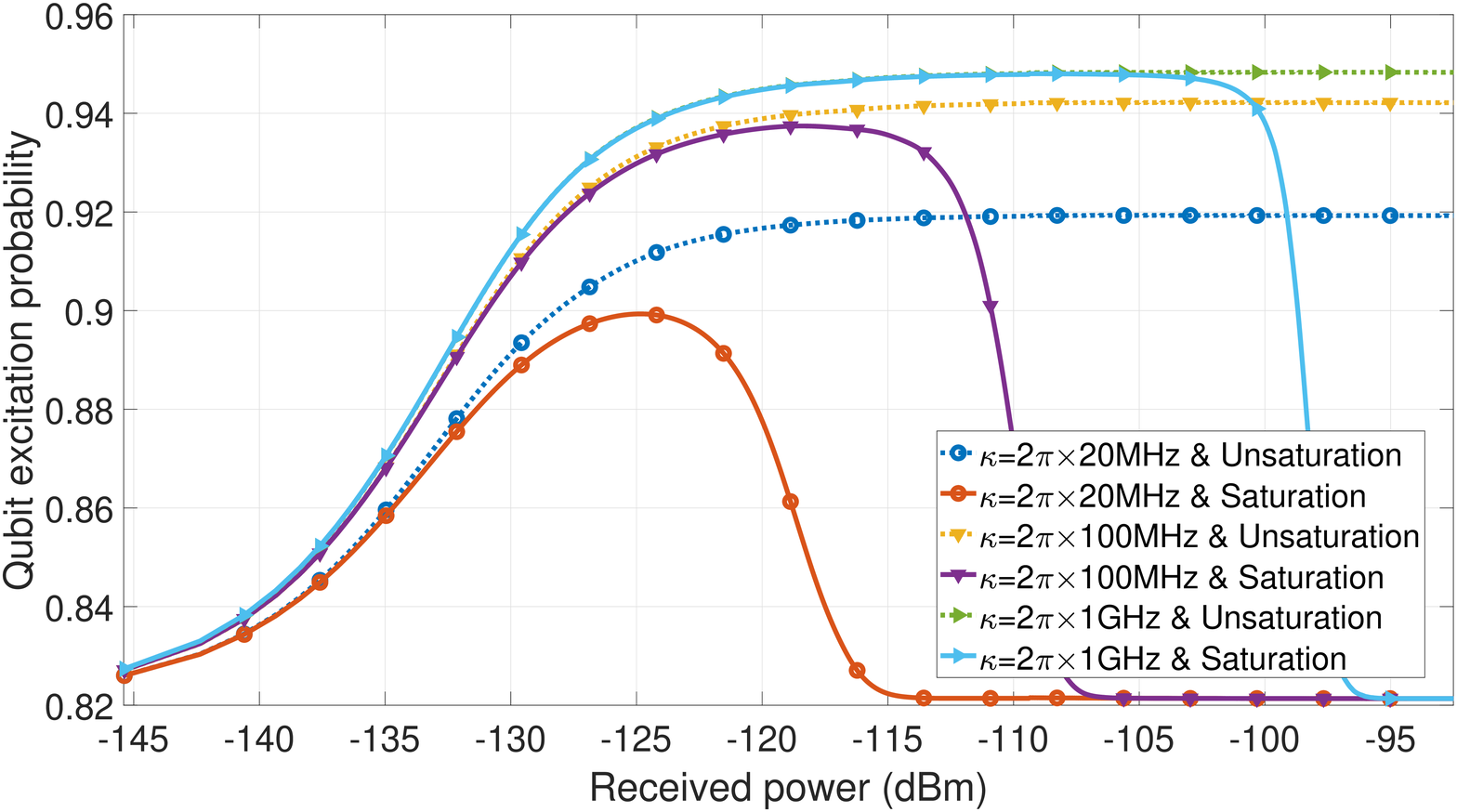}
  \caption{The qubit excitation probability under wrong reset and different transition rates $\kappa$.}
  \label{sa2to1re_fig}
\end{figure}

In the low received power regime, the saturation effect is not obvious. In this case, the gap between considering saturation and not considering saturation is negligible. For large transition rate $\kappa$, the excitation rate decreases after a flat regime. In this case, compared with the model without saturation, the maximum communication performance of the three-level system hardly decreases because the dynamic range of the qubit excitation rate is only slightly reduced.

Using the qubit excitation rate under the saturation model and the parameters in Section \ref{sec_rate} on the three-level HMM, we can obtain the achievable transmission rate, as shown in Fig. \ref{ratezu_fig}. It can be seen that when the transition rate $\kappa$ is large, considering the saturation model has almost no effect on the achievable rate.

\begin{figure}[htbp]
  \setlength{\abovecaptionskip}{-0.2cm} 
  \setlength{\belowcaptionskip}{-2cm}
  \centering
  \includegraphics[width=1\columnwidth]{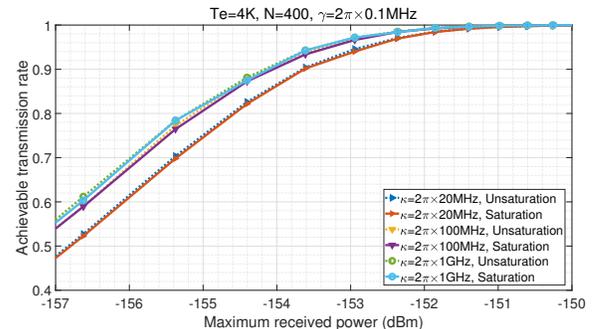}
  \caption{Achievable transmission rate under different transition rates $\kappa$.}
  \label{ratezu_fig}
\end{figure}

\subsection{Saturation Cut-off Point Characteristics}

Based on the saturation characteristics of the three-level system in Section \ref{sec_zu_rate}, we define the cutoff point at the point where the qubit excitation rate drops by 3dB after the maximum value. We further define the cutoff photon number as the mean photon number at the saturation cutoff point. The cutoff photon number under different decay rates $\kappa$ and capture time $T_c$ is shown in Fig. \ref{cutN_fig}. The number of cut-off photons under different decay rate $\gamma$ is shown in Fig. \ref{cutgamma_fig}.

\begin{figure}[htbp]
  \setlength{\abovecaptionskip}{-0.2cm} 
  \setlength{\belowcaptionskip}{-2cm}
  \centering
  \includegraphics[width=1\columnwidth]{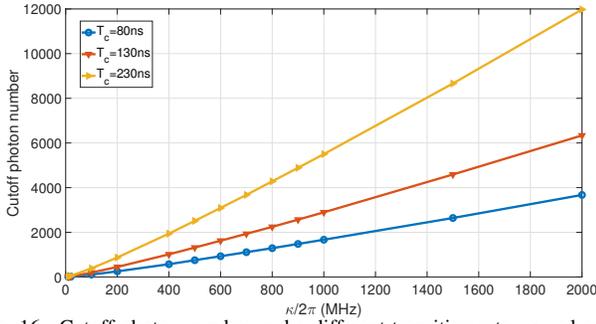}
  \caption{Cutoff photon number under different transition rates $\kappa$ and capture time $T_c$.}
  \label{cutN_fig}
\end{figure}

\begin{figure}[htbp]
  \setlength{\abovecaptionskip}{-0.2cm} 
  \setlength{\belowcaptionskip}{-2cm}
  \centering
  \includegraphics[width=1\columnwidth]{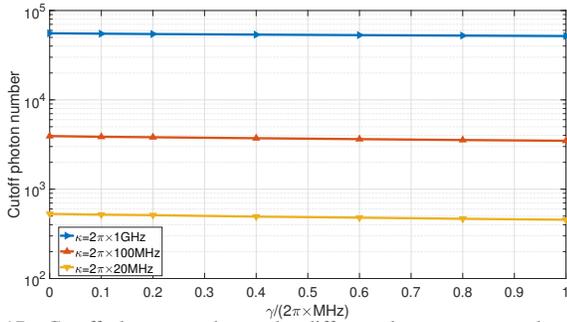}
  \caption{Cutoff photon number under different decay rates $\gamma$ and transition rates $\kappa$.}
  \label{cutgamma_fig}
\end{figure}

It can be seen that $\gamma$ has little effect on the cut-off point in the common range of $\gamma$. On the other hand, $\kappa$ determines the normalized saturation time window number, and further determines the cut-off photon number to a large extent. Plotted against $\kappa T_c$ in $\gamma=0$, the cutoff points under different values of $T_c$ overlap very well, as shown in Fig. \ref{cutkt_fig}. We further use $\kappa T_c$ as the independent variable for fitting, and the fitting result is shown in Fig. \ref{cutkt_fig}. When $\kappa T_c$ is large, the fitting result is very close to the simulation result. The fitted equation is given by,
\begin{equation}n_{\text {cutoff}}=1.457\left(\mathrm{kT}_{c}\right)^{1.132}-0.8766.\end{equation}

\begin{figure}[htbp]
  \setlength{\abovecaptionskip}{-0.2cm} 
  \setlength{\belowcaptionskip}{-2cm}
  \centering
  \includegraphics[width=1\columnwidth]{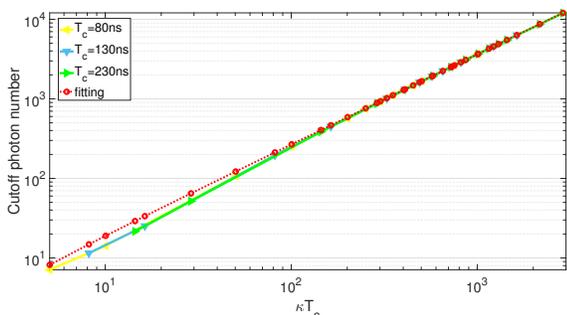}
  \caption{Cutoff photon number with $\kappa T_c$ as the horizontal axis and the fitting result.}
  \label{cutkt_fig}
\end{figure}

\section{CONCLUSION}\label{sec6}
We have adopted artificial $\Lambda$-type three level system with superconducting devices for microwave signal detection. Based on the state transition principles, we have proposed a statistical model for microwave signal detection. We have also investigated the achievable transmission rate and signal detection based on the statistical model. It is predicted that high detection sensitivity can be achieved by the proposed structure. We further studied the saturation characteristics of the three-level system and showed that it has no effect on the very weak microwave communication.

The three-level system has great application potential in the field of outer space communication due to significantly high sensitivity at low temperature, for example, deployed on spacecraft or satellites. Higher transmission rate can be achieved via deploying parallel three-level systems centered at separated frequencies.

In future work, we will study the communication signal processing of three-level systems, including symbol synchronization and the saturation problem under high received power. In addition, receiver device fabrication and real experiments is another major endeavor.

\appendix\label{app}
\subsection{Proof of Theorem \ref{thm1}}\label{app1}
We have the following
\begin{equation}E_S(N)=\mathbb{E}\left(\sum_{i=1}^{N} I_{i}\right)=N \mathbb{E}\left(I_{1}\right),\end{equation}

$\mathbb{E}\left(I_{1}\right)$ represents the probability of survival of a uniformly distributed photon. So $\mathbb{E}\left(I_{1}\right)$ is given by
\begin{equation}\begin{aligned}\mathbb{E}\left(I_{1}\right)=&\frac{2}{T_c} \int_{0}^{\tau}\left(1-\frac{t+\tau}{T_c}\right)^{N-1} d t\\&+\frac{T_c-2 \tau}{T_c}\left(1-\frac{2 \tau}{T_c}\right)^{N-1}.\end{aligned}\end{equation}
Thus we have that
\begin{equation}\begin{aligned}E_S(N)&=N \mathbb{E}\left(I_{1}\right)\\&=2\left(1-\frac{\tau}{T_c}\right)^{N}+(N-2)\left(1-\frac{2 \tau}{T_c}\right)^{N}.\end{aligned}\end{equation}

And the second moment is given by
\begin{equation}\begin{aligned}D_{S}(N)&=\mathbb{E}\left[\left(\sum_{i=1}^{N} I_{i}\right)^{2}\right]\\&=N \mathbb{E}\left(I_{1}^{2}\right)+N(N-1) \mathbb{E}\left(I_{1} I_{2}\right).\end{aligned}\end{equation}
Obviously we have
\begin{equation}\begin{aligned}N \mathbb{E}\left(I_{1}^{2}\right)&=N \mathbb{E}\left(I_{1}\right)\\&=2\left(1-\frac{\tau}{T_c}\right)^{N}+(N-2)\left(1-\frac{2 \tau}{T_c}\right)^{N}.\end{aligned}\end{equation}

$\mathbb{E}\left(I_{1} I_{2}\right)$ represents the probability that two independent photons that obey a uniform distribution are surviving photons. Assume that the arrival times of the two photons are $t_1$ and $t_2$ respectively, and $t_1\le t_2$. According to the value of $t_1$, we divide $\mathbb{E}\left(I_{1} I_{2}\right)$ into four parts, i.e.,
\begin{equation}\mathbb{E}\left(I_{1} I_{2}\right)=\frac{2}{T^2_c}(A+B+C+D).\end{equation}

\begin{itemize}
\item [1)] $0\le t_1<\tau$
\begin{equation}\begin{aligned}A=&2\int_{0}^{\tau} d t_{1} \int_{t_{1}+\tau}^{t_{1}+2 \tau} d t_{2}\left(1-\frac{t_{2}+\tau}{T_c}\right)^{N-2}\\&+\int_{0}^{\tau}d t_{1}\int_{t_{1}+2 \tau}^{T_c-\tau} d t_{2}\left(1-\frac{t_{1}+3 \tau}{T_c}\right)^{N-2}.\end{aligned}\end{equation}

Further, we have
\begin{equation}\begin{aligned}\frac{2}{T^2_c}A=&\frac{4}{N(N-1)}\left(\frac{T_c-2 \tau}{T_c}\right)^{N}\\&+\frac{2N-10}{N(N-1)}\left(\frac{T_c-3 \tau}{T_c}\right)^{N}\\&+\frac{6-2N}{N(N-1)}\left(\frac{T_c-4 \tau}{T_c}\right)^{N}.\end{aligned}\end{equation}

\item [2)] $\tau\le t_1<T_c-3\tau$
\begin{equation}\begin{aligned}B=&2\int_{\tau}^{T_c-3 \tau} d t_{1} \int_{t_{1}+\tau}^{t_{1}+2 \tau} d t_{2}\left(1-\frac{t_{2}-t_{1}+2 \tau}{T_c}\right)^{N-2}\\&+\int_{\tau}^{T_c-3 \tau} d t_{1}\int_{t_{1}+2 \tau}^{T_c-\tau} d t_{2}\left(1-\frac{4 \tau}{T_c}\right)^{N-2}.\end{aligned}\end{equation}

Further, we have
\begin{equation}\begin{aligned}\frac{2}{T_c^2}B=&\frac{4}{N-1}\left(1-\frac{4 \tau}{T_c}\right)\left(1-\frac{3 \tau}{T_c}\right)^{N-1}\\&+\frac{N-5}{N-1}\left(1-\frac{4 \tau}{T_c}\right)^{N}.\end{aligned}\end{equation}

\item [3)] $T_c-3\tau\le t_1<T_c-2\tau$
\begin{equation}\begin{aligned}C=&2\int_{T_c-3 \tau}^{T_c-2 \tau} d t_{1} \int_{t_{1}+\tau}^{T_c-\tau} d t_{2}\left(1-\frac{t_{2}-t_{1}+2 \tau}{T_c}\right)^{N-2}\\&+\int_{T_c-3 \tau}^{T_c-2 \tau} d t_{1}\int_{T_c-\tau}^{t_{1}+2 \tau} d t_{2}\left(\frac{t_{1}-\tau}{T_c}\right)^{N-2}.\end{aligned}\end{equation}

Further, we have
\begin{equation}\begin{aligned}\frac{2}{T_c^2}C=&\frac{6N\tau-6T_c}{N(N-1)T_c}\left(1-\frac{3 \tau}{T_c}\right)^{N-1}\\&+\frac{6}{N(N-1)}\left(1-\frac{4 \tau}{T_c}\right)^{N}.\end{aligned}\end{equation}

\item [4)] $T_c-2\tau\le t_1<T_c-\tau$
\begin{equation}D=\int_{T_c-2 \tau}^{T_c-\tau} d t_{1} \int_{t_{1}+\tau}^{T_c} d t_{2}\left(1-\frac{T_c-t_{1}+\tau}{T_c}\right)^{N-2}.\end{equation}

Further, we have
\begin{equation}\begin{aligned}\frac{2}{T_c^2}D=&\frac{2}{N(N-1)}\left[\left(1-\frac{2 \tau}{T_c}\right)^{N}-\left(1-\frac{3 \tau}{T_c}\right)^{N}\right]\\&-\frac{2\tau}{(N-1)T_c}\left(1-\frac{3 \tau}{T_c}\right)^{N-1}.\end{aligned}\end{equation}
\end{itemize}

According to the above results, the second moment is given by
\begin{equation}\begin{aligned}D_{N}=&2\left(1-\frac{\tau}{T_c}\right)^{N}+(N+4)\left(1-\frac{2 \tau}{T_c}\right)^{N}\\&+\left(N^{2}-7 N+12\right)\left(1-\frac{4 \tau}{T_c}\right)^{N}\\&+(6 N-18)\left(1-\frac{3 \tau}{T_c}\right)^{N}.\end{aligned}\end{equation}

\subsection{Proof of Theorem \ref{thm3}}\label{app2}
We define $a=\lambda\tau$, $\beta=T_c/\tau\ge4$ and $S(a)=\Delta_\lambda e^{4a}$. And we have the following
\begin{equation}\begin{aligned}S(a)=&-2 e^{2 a}-[(2 \beta-10) a-10] e^{a}+(4 \beta-12) a^{2}\\&+(2 \beta-16) a-8,\end{aligned}\end{equation}
\begin{equation}\begin{aligned}\frac{\mathrm{d} S}{\mathrm{d} a}=&-4 e^{2 a}-[(2 \beta-10)(a+1)-10] e^{a}\\&+(8 \beta-24) a+(2 \beta-16),\end{aligned}\end{equation}
\begin{equation}\frac{\mathrm{d}^2 S}{\mathrm{d} a^2}=-8 e^{2 a}-[(2 \beta-10)(a+2)-10] e^{a}+(8 \beta-24),\end{equation}
\begin{equation}\begin{aligned}\frac{\mathrm{d}^k S}{\mathrm{d} a^k}&=-2^{k+1} e^{2 a}-[(2 \beta-10)(a+k)-10] e^{a}\\&\triangleq p_{k}(a) e^{a}\end{aligned},\ k \geq 3.\end{equation}

When $\lambda=0$, we have $S(0)=0$, $S^{(1)}(0)=0$ and $S^{(2)}(0)=4\beta-2>0$. In order to illustrate the value of $S$ with $a>0$, we further prove some theorems of $p_k(a)$.
\begin{lemma}\label{thm4}
If $k\ge1$ and $p_k(0)\le0$, we have $p_k(a)\le0$ with $a\ge0$.
\end{lemma}
\begin{proof}
we have
\begin{equation}p_{k}(0)=-2^{k+1}-(2 M-10) k+10 \leq 0,\end{equation}
\begin{equation}\begin{aligned}p_{k}(a)-p_{k}(0)&=-2^{k+1}\left(e^{a}-1\right)-(2 M-10) a \\&\leq-\left(2^{k+1}+2 M-10\right)a.\end{aligned}\end{equation}
If $2M-10\le0$, we have the following
\begin{equation}-2^{k+1}-(2M-10)<-2^{k+1}-(2M-10)k+10\le0,\end{equation}
\begin{equation}p_{k}(a)-p_{k}(0)\leq-\left(2^{k+1}+2 M-10\right)a\le0.\end{equation}
If $2M-10>0$, $p_k(a)-p_k(0)\le0$ obviously.
Thus, we prove
\begin{equation}p_{k}(a)=p_k(0)+p_{k}(a)-p_k(0)\le0.\end{equation}
\end{proof}

Because of $\lim_{k\to\infty}p_k(0)\to-\infty$ and $\lim_{a\to+\infty}p_k(a)\to-\infty$, we let $k^*$ be the smallest positive integer such that $p_k(a)<0$. If $k^*\ge4$, $S^{(3)}$ is positive first and then negative as $a$ increases. When $k<4$, $S^{(3)}\le0$ always holds. In other words, $S,\ S^{(1)}$ and $S^{(2)}$ are all positive first and then negative with $a>0$.

\small{\baselineskip = 10pt
	\bibliographystyle{IEEEtran}
    \bibliography{ustc}

\end{document}